\documentclass[12pt]{iopart}

\usepackage{enumerate}
\usepackage{amssymb}

\usepackage{amsthm}
\newtheorem{thm}{Theorem}

\newcommand{\be}{\begin{equation}}
\newcommand{\ee}{\end{equation}}
\newcommand{\ba}{\begin{eqnarray}}
\newcommand{\ea}{\end{eqnarray}}
\begin{document}

\title[Conformal symmetry classes for pp-wave spacetimes]
{Conformal symmetry classes for pp-wave spacetimes}

\author{Aidan J Keane$^1$ and Brian O J Tupper$^2$}
\address{$^1$ Department of Physics and Astronomy, University of Glasgow, \\
Glasgow G12 8QQ, Scotland, UK\\
$^2$ Department of Mathematics and Statistics, University of
New Brunswick\\Fredericton, New Brunswick, E3B 5A3, Canada}

\begin{abstract}
We determine conformal symmetry classes for the pp-wave spacetimes. This refines the isometry classification scheme given by Sippel and Goenner (1986 {\it Gen. Rel. Grav.} {\bf 18} 1229). It is shown that every conformal Killing vector for the null fluid type $N$ pp-wave spacetimes is a conformal Ricci collineation. The maximum number of proper non-special conformal Killing vectors in a type $N$ pp-wave spacetime is shown to be three, and we determine the form of a particular set of type $N$ pp-wave spacetimes admitting such conformal Killing vectors. We determine the conformal symmetries of each type $N$ isometry class of Sippel and Goenner and present new isometry classes.
\end{abstract}

%Uncomment for PACS numbers title message
% Keywords required only for MST, PB, PMB, PM, JOA, JOB?
%\vspace{2pc}
%\noindent{\it Keywords}: Article preparation, IOP journals
% Uncomment for Submitted to journal title message
%\submitto{\CQG}
% Comment out if separate title page not required
%\maketitle

\section{Introduction}
\label{sec:intro}
We define a {\it pp-wave spacetime} $M$ to be a non-flat spacetime
which admits a covariantly constant, nowhere zero, null bivector.
This is equivalent to $M$ admitting a global, covariantly constant, nowhere zero, null vector and either
(i) $M$ having Weyl tensor of Petrov type $N$ or $O$ at each $p \in M$, or
(ii) $M$ having Ricci tensor either zero or Segre type $\{ (211) \}$ with zero eigenvalue at each $p \in M$ \cite{hall90}.
We note that in \cite{kramer80} a pp-wave is defined to be a spacetime admitting only a global, covariantly constant, nowhere zero, null vector.
In the following we shall use the notation $x^a \equiv (u,v,y,z) \equiv (u, v, x^A)$ where
$a=0,1,2,3$ and $A=2,3$. The notation ${\cal A} \supset {\cal B}$ means ${\cal B}$ is a subalgebra of ${\cal A}$.
The line element for the pp-wave spacetime can be written \cite{ehlers62}
\be
ds^2 =  - 2 du dv - 2 H(u,x^A) du^2 + \delta_{AB}dx^A dx^B \: ,
\label{eqn:ppw}
\ee
and the null covariantly constant vector $k$ is necessarily a Killing vector and has the form
\be
k^a = \delta^a_v \; , \qquad k_a = - \delta^u_a \: .
\label{eqn:kvk}
\ee
Sippel and Goenner \cite{sippel86} give the form of the Riemann, Weyl and Ricci tensors for the pp-wave spacetime, the latter being
$R_{ab} = F k_a k_b$ where $F = H_{,yy} + H_{,zz}$.
Ehlers and Kundt \cite{ehlers62} give the form of the covariantly constant null bivector.
If the Weyl tensor is nowhere zero then $R_{abcd} k^d = 0$ and $k$ is a repeated principal null direction of
the Weyl tensor \cite{hall90}.
The spacetime (\ref{eqn:ppw}) is vacuum if $H_{,yy} + H_{,zz} = 0$ and conformally flat if $H_{,yy} = H_{,zz}$ and $H_{,yz}=0$.
Further, if the metric function can be put in the form
\be
2H = A(u) y^2 + 2B(u) yz + C(u) z^2 \: ,
\label{eq:H}
\ee
then the spacetime is referred to as a {\it plane wave spacetime}.

Let $M$ be a four-dimensional spacetime manifold with metric tensor $g$ of Lorentz
signature. Any vector field $Y$ which satisfies
\be
{\cal L}_Y \,  g = 2 \psi(x^\alpha) \, g \label{eqn:ckv}
\ee
is said to be a {\it conformal Killing vector} (CKV) of $g$.
If $\psi$ is not constant on $M$ then $Y$ is called a {\it proper
conformal Killing vector}. If $\psi_{;ab} = 0$, then $Y$ is
called a {\it special conformal Killing vector} (SCKV).
If $\psi$ is constant on $M$ then $Y$ is
called a {\it  homothetic Killing vector} (HKV)
and if $\psi=0$ then $Y$ is said to be a
Killing vector (KV).
An SCKV is called a proper SCKV if $\psi_{;a} \ne 0$ and an HKV is called a proper HKV if
$\psi \ne 0$.
The set of all CKV (respectively, SCKV, HKV and KV) form a finite-dimensional Lie algebra
denoted by ${\cal C}$ (respectively, ${\cal S}$, ${\cal H}$ and ${\cal G}$).
The maximum dimension of the algebra of CKV on $M$ is 15 and this is
achieved if $M$ is conformally flat. If the spacetime is not
conformally flat, then the maximum dimension is seven \cite{hall91}, and this occurs for
special type $N$ models \cite{hall91}, \cite{defrise75}. The spacetimes in question are necessarily conformally related to a  plane wave spacetime, which in turn is necessarily
conformally related to the vacuum plane wave spacetime \cite{hall90}.
To be specific, a type $N$ spacetime with a seven-dimensional conformal algebra is conformally related to a plane wave spacetime with ${\cal H}_7$.
Further, a type $N$ spacetime with a six-dimensional conformal algebra is conformally related to one of two special pp-wave spacetimes.
By the Defrise-Carter theorem \cite{defrise75, hall91} the spacetime is conformally reducible to either a type $N$
plane wave with an ${\cal H}_6$ or a non-plane wave type $N$ spacetime with a ${\cal G}_6$.
There is only one type $N$ spacetime with a ${\cal G}_6$ (metric (12.6) in \cite{kramer80}) and it takes the form
\be
ds^2 = y^{-2} (- 2 du dv - 2 y^{-2} du^2 + dy^2 + dz^2) \, .
\label{eq:defriseG6}
\ee
However, this spacetime is conformally related to a particular type of pp-wave with an ${\cal S}_6 \supset {\cal H}_5$,
i.e., class $Biv$ of section \ref{sec:typeN}.
In this paper we investigate the conformal symmetry properties of the entire class of pp-wave spacetimes.
We determine the Lie algebra structures, which are unaffected by a conformal rescaling of
the metric.

Ehlers and Kundt \cite{ehlers62} determined the isometry classes of pp-wave spacetimes satisfying the vacuum field
equations, Sippel and Goenner \cite{sippel86} determined the general form for the KV of a general pp-wave spacetime (i.e., without the vacuum
restriction) and generalized the classification in \cite{ehlers62} to include non-vacuum pp-wave spacetimes.
Maartens and Maharaj \cite{maartens91} determined the general form for the CKV of a pp-wave spacetime and, in particular,
give expressions for the HKV and SCKV in such spacetimes. They also give examples of non-special CKV.
Note the difference in signature of the metrics and sign difference
in the definition of $H$ in \cite{sippel86} and \cite{maartens91}.
Hall {\it et al} \cite{hall92} determined the CKV for the
general conformally flat pp-wave spacetime. The conformal symmetries of the
class of non-twisting type $N$ vacuum solutions (with and without a cosmological constant) are listed in \cite{humberto83},
see also \cite{garcia81}. Siklos \cite{siklos85} determines the KVs of the Lobatchevski plane wave spacetimes,
which are conformally related to certain pp-wave spacetimes. Aichelburg and Balasin \cite{aichelburg96}
determine the KV of pp-wave spacetimes with distributional profiles, and in \cite{aichelburg97} determine a set of
generalized symmetries for such spacetimes.

The conformal symmetries of the type $O$ pp-wave spacetimes
can be dealt with in a straightforward manner.
We can represent any such spacetime in a conformally flat coordinate system:
Consider Minkowski spacetime in coordinates $ds^2 = -2dudv + dy^2 + dz^2$,
then
\[
d{\bar s}^2 = \Omega^2(u) \,  (-2dudv + dy^2 + dz^2)
\]
is, from theorem 3 of section 4 (see also \cite{tupper03}), a plane wave spacetime with line element
\[
d{\bar s}^2 = \Omega^{-1}(\bar u) \Omega(\bar u)_{,\bar u \bar u} (\bar y^2 + \bar z^2) d\bar u^2 -2d \bar u d \bar v + d \bar y^2 + d \bar z^2.
\]
The conformal algebra is simply the standard $so(4,2)$ conformal algebra of Minkowski spacetime.

It remains for us to consider the conformal symmetry properties of the type $N$ pp-wave spacetimes. Only vacuum and null fluid
spacetimes are possible, and the following statements are useful.
The field equations for a null fluid are (cosmological constant $\Lambda=0$)
\be
G_{ab} = R_{ab} = F \, k_a k_b \; ,
\label{eq:nullfluid}
\ee
where $F$ is a non-zero function on $M$, $k_a k^a = 0$ and so $R={R^a}_a = 0$. We shall make use of the following theorem.

\begin{thm} \label{thm1}
Every CKV in a non-conformally flat null fluid spacetime with Ricci tensor (\ref{eq:nullfluid}), $k$ being one of the principal null directions
of the Weyl tensor, is a conformal Ricci collineation and $\psi_{;ab}$ is proportional to $k_a k_b$.
\end{thm}

\begin{proof}
From (\ref{eq:nullfluid}) it follows that
${\cal L}_Y \, R_{ab} = ({\cal L}_Y F) \, k_a k_b + F {\cal L}_Y (k_a k_b)$.
Since ${\cal L}_Y k_a = \alpha k_a$ for some function $\alpha$ \cite{hall91}, it follows that the CKV satisfies
${\cal L}_Y R_{ab} = 2 \eta R_{ab}$, $2 \eta = {\cal L}_Y (\ln F) + 2\alpha$,
that is, $Y$ is a conformal Ricci collineation.
Further, it can be shown that \cite{yano55}
\ba
{\cal L}_Y R_{ab} & =  -2 \psi_{;ab} - g_{ab} \square \psi , \label{eq:riccit}
\\
{\cal L}_Y R & =  -2 \psi R - 6 \square \psi \: \: , \label{eq:riccis}
\ea
where $\square \psi = g^{ab} \psi_{;ab}$. Equation (\ref{eq:riccis}) then gives $\square \psi=0$ and then (\ref{eq:riccit}) gives
${\cal L}_Y R_{ab} = -2  \psi_{;ab}$. It follows that $\psi_{;ab} = - \eta \, F \, k_a k_b$.

\end{proof}

Maartens and Maharaj \cite{maartens91} give the general form for the CKV of a
general (type $N$ or type $O$) non-flat pp-wave spacetime as (here and throughout a prime denotes differentiation with respect to $u$)
\ba
Y^u & = & \mu \delta_{AB} x^A x^B / 2 + a_A(u) x^A + a(u) \: , \nonumber\\
Y^v & = & \mu v^2 + [a'_A(u) x^A + 2b(u) -a'(u)] v +F(u, x^A) \: , \nonumber\\
Y^A & = & [\mu x^A + a_A(u)] v + \gamma_{ABCD} a'_B(u) x^C x^D + b(u) x^A + c(u) \epsilon_{AB} x^B + c_A(u) \: , \nonumber
\ea
and the conformal scalar $\psi$ and $\psi_{;ab}$ as
\ba
\fl \psi =  \mu v + a'_A(u) x^A + b(u) \: , \\
\fl \psi_{;ab} = (a'''_A(u) x^A + b''(u) - \mu H_{,u} -a'_A(u) H_{,A}) k_a k_b
+ 2 (a''_B(u) - \mu H_{,B}) k_{(a} {x^B}_{,b)} \: , \label{eq:psiab}
\ea
where $\gamma_{ABCD}=\delta_{AB} \delta_{CD}/2 - \epsilon_{AC} \epsilon_{BD}$, $\mu$ is a constant,
$a$, $b$, $c$, $a_A$, $c_A$ and $F$ are functions and the following conditions are satisfied
(equations (16) and (17) in  \cite{maartens91})
\ba
{Y^v}_{,A} + 2H {Y^u}_{,A} - {Y^A}_{,u} & = & 0 \: , \label{eq:Hcondition1}\\
{Y^v}_{,u} + H_{,a} Y^a + 2H {Y^u}_{,u} & = & 2 \psi H \: . \label{eq:Hcondition2}
\ea
Extra conditions are imposed on $H$ by the integrability conditions $F_{,23}=F_{,32}$ and $F_{,uA}=F_{,Au}$, which are equations
(36) and (37) in  \cite{maartens91}. Note that $4 \psi = {Y^a}_{;a} = {Y^a}_{,a}$ for a pp-wave spacetime.
The general form for the SCKV in the pp-wave spacetime is given by \cite{maartens91}
\be
W = \rho S + \phi Z + X
\label{eq:generalSCKV}
\ee
where
\ba
\fl S & = & u^2 \partial_u + \case{1}{2} \delta_{AB} x^A x^B \partial_v
+ u \, x^A \partial_A \: , \\
\fl Z & = & 2 v \partial_v
+ x^A \partial_A \: , \label{eq:vectorZ}\\
\fl X & = & (\alpha u + \beta) \partial_u +
(\lambda(u) - \alpha v + c'_A(u) x^A) \partial_v +
(\gamma \epsilon_{AB} x^B + c_A(u)) \partial_A \: ,
\ea
with corresponding functions $a(u) = \rho u^2 + \alpha u + \beta$ and $\psi = \rho u + \phi$, where
$\rho$, $\alpha$, $\beta$, $\gamma$ and $\phi$ are constants. In \cite{maartens91} it is stated that $\lambda(u)$ can be
replaced by a constant using equation (10) of \cite{maartens91}, however we have found examples where this is not the case
e.g., the HKV (\ref{eq:case6HKV}) in section \ref{sec:SandG}.
Note that the most general form for an HKV is given by $\phi Z + X$ and the most general form for a KV is given by $X$.
However, if a spacetime admits an SCKV given by (\ref{eq:generalSCKV}), this {\it does not} necessarily mean that the vector field
$\phi Z + X$ is an HKV or that the vector field $X$ is a KV.
For a general spacetime, $\dim {\cal H}$ is at most one greater than $\dim {\cal G}$ and for an arbitrary non-flat spacetime
$\dim {\cal S}$ is at most one greater than $\dim {\cal H}$ \cite{hall90ii}.

We aim to give a classification of the pp-wave spacetimes according to their conformal symmetries.
This is greatly simplified by the fact that the Petrov type is $N$ or $O$.
Further simplification is allowed by the fact that the energy momentum tensor is either zero or that of a
null fluid spacetime, in the former case only SCKV are possible (from equations (\ref{eq:riccit}) and (\ref{eq:riccis}))
and in the latter we make use of theorem 1 above.
In section \ref{sec:typeN} it is shown that the form of the general CKV in a type $N$ pp-wave spacetime
can be restricted further than shown in \cite{maartens91}.
The maximum number of non-special CKVs in a type $N$ pp-wave spacetime is shown to be three.
In section \ref{sec:typeN} we examine a particular class of such spacetimes which admit three non-special CKVs. Further, we identify
additional isometry classes, not appearing in \cite{sippel86}.
In section \ref{sec:SandG} we determine the form of the general CKV in
each of the classes 1 - 14 of \cite{sippel86} and attempt to enumerate all possible Lie algebra structures. Some new isometry classes are introduced.
It transpires that none of the Sippel and Goenner spacetime classes, except those in class 10 (which admits an ${\cal H}_6$)
and its specializations (which admit an ${\cal H}_6$ subalgebra), automatically admits any further CKV.
However, further CKV are admitted in certain cases when the
functional form of $H$ is further restricted. Isometry classes 5-14 are resolved completely.
However, classes 1-4 prove to be troublesome on account of the fact that the condition
(\ref{eq:Hsatisfy}) leads to partial differential equations in general.
In section \ref{sec:typeNdual} we present a set of conformally related pp-wave spacetimes and determine their conformal
symmetries.

\section{CKV for type $N$ pp-wave spacetimes}
\label{sec:typeN}
Let us now specialize to the type $N$ spacetimes.
It follows from theorem 1 that for an arbitrary CKV in a type $N$ null fluid pp-wave spacetime, $\psi_{;ab}$ is {\it always} proportional to $k_a k_b$. (Of course, for the vacuum spacetimes $\psi_{;ab}$ is identically zero.) Thus the last bracketed term appearing in equation (\ref{eq:psiab}) is zero for all CKV, that is,
\be
a''_A(u) - \mu H,_A = 0 \: .
\ee
Since we only wish to consider non-flat pp-wave spacetimes then we must have
$\mu = 0$ and $a''_A(u)=0$ \cite{maartens91}.
Equations (33) and (36) of \cite{maartens91} then force $a_A(u)=0$, $a''(u) = 2b'(u)$ and $c'(u)=0$.
Thus
\be
a'(u) =  2 b(u) + \Theta,
\label{eq:adashed}
\ee
where $\Theta$ is a constant and $c=\gamma$ is a constant.
Then the components of the CKV have the form
\ba
Y^u & = & a(u) \: , \nonumber\\
Y^v & = & -\Theta v + a''(u) \, \delta_{AB} x^A x^B/4 + c'_B(u) \, x^B + E(u) \: , \nonumber\\
Y^A & = & (a'(u) - \Theta) \, x^A / 2 + \gamma \, \epsilon_{AB} x^B + c_A(u) \: , \label{eq:refinedCKV}
\ea
where $E$ is a function of $u$ and the conformal scalar $\psi$ is a function of $u$ only
\be
\psi = b(u)=(a'(u) - \Theta)/2 \; .
\label{eq:newpsi}
\ee
Further
\be
\psi_{;ab} = \psi'' k_a k_b = a'''(u) k_a k_b / 2 \: .
\label{eq:psiab2}
\ee
Equation (\ref{eq:Hcondition1}) is satisfied identically by (\ref{eq:refinedCKV}) and equation (\ref{eq:Hcondition2}) becomes
\be
H_{,a} Y^a + {Y^v}_{,u} = - (a'(u) + \Theta) \, H \: .
\label{eq:Hsatisfy}
\ee
Further, the integrability conditions $F_{,23}=F_{,32}$ and $F_{,uA}=F_{,Au}$ are satisfied identically.
Note that for SCKV, the constant $\phi$ in equation (\ref{eq:generalSCKV}) is given by $\phi = (\alpha - \Theta)/2$.

For the moment, let us convert to polar coordinates $y = r \cos \theta$ and $z = r \sin \theta$. The metric (\ref{eqn:ppw}) becomes
\be
ds^2 = -2 du dv -2H(u, r, \theta) \, du^2 + dr^2 + r^2 d \theta^2 \, ,
\ee
and, relabelling $c_y(u)$, $c_z(u)$ as $c_2(u)$, $c_3(u)$, respectively, the CKV (\ref{eq:refinedCKV}) becomes
\ba
Y^u & = & a(u) \: , \nonumber\\
Y^v & = & -\Theta v + a''(u) \, r^2 / 4 + c'_2(u) r \cos \theta + c'_3(u) r \sin \theta + E(u) \: , \nonumber\\
Y^r & = & (a'(u) - \Theta) \, r / 2 + c_2(u) \cos \theta + c_3(u) \sin \theta \: , \nonumber\\
Y^{\theta} & = & - \gamma + ( - c_2(u) \sin \theta + c_3(u) \cos \theta ) / r \: .
\label{eq:refinedCKVpolars}
\ea
Equation (\ref{eq:Hsatisfy}) becomes
\be
a'''(u) + L(u ,r ,\theta) \, a'(u) + M(u, r, \theta) \, a(u) = N (u, r, \theta) \, ,
\label{eq:thirdorder}
\ee
where
\ba
\fl L = 4 r^{-2} (H + r H_{,r} / 2)  \, , \qquad M = 4 r^{-2} H' \, , \nonumber\\
\fl N = 4 r^{-2} [(- H + r H_{,r} / 2) \, \Theta - r (c''_2(u) \cos \theta + c''_3(u) \sin \theta)
-H_{, r} (c_2(u) \cos \theta + c_3(u) \sin \theta) \nonumber\\
- H_{, \theta} (- \gamma + (-c_2(u) \sin \theta + c_3(u) \cos \theta)r^{-1})  - E'(u) ] \, .
\nonumber
\ea
Each CKV, $Y$, is characterized by a set of four functions of $u$ and two constants, i.e., $(a(u), c_A(u), E(u), \gamma, \Theta)$ and, once $H$
has been chosen, the differential equation (\ref{eq:thirdorder}) governs their behaviour. Fixing the function $H$ fixes $L$ and $M$ but not $N$,
in which case the lhs of the differential equation is known completely. $N$ depends on $c_A(u)$, $E(u)$, $\gamma$ and $\Theta$. Note that
the lhs will only depend upon $\theta$ if the function $H$ depends on $\theta$ and that if $H$ is independent of $\theta$ then $c_A(u) = 0$.
In general, equation (\ref{eq:thirdorder}) will separate into equations for each of $u$, $r$ and $\theta$.
In fact, the maximum number of non-special CKV admitted by a pp-wave spacetime can be deduced from (\ref{eq:thirdorder}).
In order for the spacetime to admit at least one non-special CKV the condition $a'''(u) \ne 0$ must be satisfied.
For one particular choice of $H$ and the parameters $(c_A(u), E(u), \gamma, \Theta)$ one can have at most three independent
non-special CKV. In such a case one might ask if there can be any further independent non-special CKV.
The answer is in the negative because the parameters $(a(u), c_A(u), E(u), \gamma, \Theta)$ appear {\it linearly} in the CKV $Y$,
that is, if there were any further such CKV one could find a linear combination of all these CKV providing a new set
of parameters $({\bar c}_A(u), {\bar E}(u), {\bar \gamma}, {\bar \Theta})$ for which the corresponding third order ordinary
differential equation for ${\bar a(u)}$ would necessarily have more than three independent solutions.

\begin{thm} \label{propthree}
A type $N$ pp-wave spacetime can admit at most three non-special CKV.
\end{thm}

Note that, if $L$, $M$ and $N$ are functions of $u$ only, equation (\ref{eq:thirdorder}) admits three independent solutions so
that the spacetime can admit the maximum number of three non-special CKV.
If $L$, $M$ and $N$ are not all functions of $u$ only, it is not clear in general whether it is possible for the spacetime
to admit three non-special CKV.
However, although not immediately apparent, there are some examples which for which $L$, $M$ and $N$ are functions of
$u$ and $r$ and admit three non-special CKV. We shall discuss these shortly.
To find solutions admitting the maximum number of non-special CKV we put
$L = \mu(u)$, $M = \nu(u)$ and $N = \eta(u)$. Equation (\ref{eq:thirdorder}) separates to give
\ba
a'''(u) + 8 \tau (u) a'(u) + 4 \tau'(u) a(u) = 0 \, , \label{eq:thirdorder2} \\
H = \tau (u) r^2 + \delta(\theta) r^{-2} \, , \label{eq:maxconditions1} \\
2 \, \delta(\theta) \, \Theta = \gamma \delta(\theta)_{, \theta} \, , \label{eq:fourteen} \\
\delta(\theta)_{, \theta} (c_2(u) \sin \theta - c_3(u) \cos \theta ) + 2 \delta(\theta) (c_2(u) \cos \theta + c_3(u) \sin \theta ) = 0 \, ,
\label{eq:fifteen}\\
c''_A(u) + 2 \tau (u) c_A(u) = 0 \, , \label{eq:sixteen}
\ea
and $E'(u) = 0$, $\eta(u) = 0$, where $4 \tau'(u) = \nu(u)$, and $\delta(\theta) \ne 0$, otherwise the spacetime is type $O$.
Four possible classes of solutions arise from equations (\ref{eq:maxconditions1}) - (\ref{eq:sixteen}), namely

\begin{tabbing}

{0000} \= {0000000000000000000000000000000000000000000000000000000000} \kill
$A$   \> $\delta (\theta) = l e^{2 m \theta}$, where $l \ne 0,m \ne 0$ constants, $\gamma m = \Theta$, $c_A = 0$. \\

$B$  \>  $\delta (\theta) = l (\sigma \sin \theta - \rho \cos \theta)^{-2}$, where $l \ne 0$, $\rho$ and $\sigma$ are constants (not both zero), \\
        \>  $\rho c_2 = \sigma c_3$, $c_A(u)$ satisfy (\ref{eq:sixteen}), $\Theta = \gamma = 0$.                       \\

$C$  \> $\delta (\theta)$ is constant, $c_A = \Theta = 0$. \\

$D$  \> $\delta (\theta)$ is arbitrary, $c_A = \Theta = \gamma = 0$.
\end{tabbing}

We note that in classes $A$, $C$ and $D$ above we can simply add an arbitrary function
of $u$ to $H$ in (\ref{eq:maxconditions1}), i.e., $H \mapsto H + f(u)$ to obtain examples where $L$, $M$ and $N$ are functions of
$u$ and $r$ and admit three non-special CKV. However, such a function of $u$ can always be transformed away by a suitable
coordinate transformation and we shall say no more about them here.

There are specializations of the above classes $A-D$ in which one or more of the CKV degenerates into a proper SCKV,
proper HKV or KV. The possible specializations are listed in table \ref{tab:maxtypes}. Consider the case in which such
a degeneration occurs. Then from equation
(\ref{eq:psiab2}), there must be at least one solution for which $a'''(u) = 0$, i.e.,
\be
a(u) = \rho u^2 + \alpha u + \beta \, ,
\label{eq:aofuforsckv}
\ee
where $\rho,\alpha, \beta$ are constants.
This gives rise to two possibilities. Either $\tau (u) = 0$, in which case $\rho, \alpha, \beta$ are arbitrary giving rise to a proper SCKV, proper HKV and KV,
or $\tau(u) \ne 0$, so that, from equation (\ref{eq:thirdorder2})
\be
\tau(u) = c (\rho u^2 + \alpha u + \beta)^{-2} \, ,
\label{eq:tauofuforsckv}
\ee
where $c$ is a non-zero constant. This implies that $\rho, \alpha, \beta$ are fixed constants, since they appear in the spacetime metric, so there is only one
SCKV which will be a proper SCKV if $\rho \ne 0$ or a HKV or KV (depending on the value of $\Theta$) if $\rho = 0$. The two remaining CKV will both be
non-special as can be seen by substituting for $\tau(u)$ given by (\ref{eq:tauofuforsckv}) into equation (\ref{eq:thirdorder2}), which becomes
\be
a'''(u) + 8c (\rho u^2 + \alpha u + \beta)^{-2} a'(u) - 8 c (2 \rho + \alpha)(\rho u^2 + \alpha u + \beta)^{-3} a(u) = 0 \, .
\ee
Apart from the solution given by (\ref{eq:aofuforsckv}), this equation admits two further solutions each of which satisfies $a'''(u) \ne 0$, namely
\ba
(\rho u^2 + \alpha u + \beta)^{-1} a(u) = \sin (\omega \Phi) \: \hbox{or} \: \sinh (\omega \Phi) \, , \nonumber\\
(\rho u^2 + \alpha u + \beta)^{-1} a(u) = \cos (\omega \Phi) \: \hbox{or} \: \cosh (\omega \Phi) \, ,
\ea
where $\Phi = \int (\rho u^2 + \alpha u + \beta)^{-1} du$ and the constant $\omega$ is related to the constants $\rho,\alpha, \beta, c$ by
\be
4 \rho \beta - \alpha^2 + 8 c = \pm \omega^2 \, .
\ee
The positive (negative) sign corresponds to the trigonometric (hyperbolic) functions.
\begin{table}
\caption{The specializations of the conformal symmetry classes $A$, $B$, $C$ and $D$. The entries correspond to the three CKV (\ref{eq:thirtysix}).}
\footnotesize\rm
\begin{tabular*}{\textwidth}{@{}l*{15}{@{\extracolsep{0pt plus12pt}}l}}
\br
& & & Non-special & & &\\
Type & $\tau(u)$ & $a'''(u)$ & CKV & pSCKV & pHKV &KV\\
\mr
General & $\ne 0$ & $\ne 0$ & $3$ & $0$  & $0$ & $0$  \\
(i) & $\ne 0$ & $\ne 0$ & $2$ & $1$  & $0$ & $0$  \\
(ii) & $\ne 0$ & $\ne 0$ & $2$ & $0$  & $1$ & $0$  \\
(iii) & $\ne 0$ & $\ne 0$ & $2$ & $0$  & $0$ & $1$  \\
(iv) & $=0$ & $=0$ & $0$ & $1$  & $1$ & $1$  \\
\br
\end{tabular*}
\label{tab:maxtypes}
\end{table}
Labelling the three solutions for $a(u)$ as $a_i(u)$, $i = 1,2,3$, the three corresponding CKV are given by
\be
Y_i = a_i(u) \partial_u + \case{1}{4} a''_i(u) r^2 \partial_v
+ \case{1}{2} a'_i(u) r \partial_r \, .
\label{eq:thirtysix}
\ee
The Lie brackets for these CKV are
\[
[ Y_i , Y_j ] = \epsilon^k_{ij} Y_k \, ,
\]
where $\epsilon^k_{ij}$ is the alternating symbol and $a_k = a_i a'_j - a_j a'_i$.
Apart from these CKV and the KV $k$, classes $A$, $B$ and $C$ above will admit additional proper SCKV, proper HKV or KV  corrsponding to the
non-zero values of $\gamma$, $\Theta$ and $c_A(u)$. However, class $D$ admits no additional symmetries.

As examples of the cases described in table \ref{tab:maxtypes}, consider the case when $a(u) = u^n$. If $n=2, 1$ or $0$ then $a'''(u) = 0$ and the
CKV given by equation (\ref{eq:thirtysix}) is, respectively, a proper SCKV, a proper HKV and a KV. For all values of $n$ other than $n=1$,
equation (\ref{eq:thirdorder2}) leads to
\be
\tau(u) = - \case{1}{8} n (n-2) u^{-2} + \omega u^{-2n}
\label{eq:october181145}
\ee
where $\omega$ is constant and the three solutions for $a(u)$ are
\ba
\fl a_i(u) = u^n, u^n \sin (q u^{1-n}), u^n \cos (q u^{1-n}) \qquad
&& \hbox{for} \: \: \omega = q^2 (1-n)^2 / 8 > 0 \label{eq:october33A} \\
\fl a_i(u) = u^n, u^n \sinh (q u^{1-n}), u^n \cosh (q u^{1-n}) \qquad
&&  \hbox{for} \: \: \omega = -q^2 (1-n)^2 / 8 < 0. \label{eq:october33B}
\ea
In general these $a_i(u)$ will lead to three non-special CKV. However, when $n=2$, i.e.,
\be
\tau(u) = \omega u^{-4}
\label{eq:october181116}
\ee
we have an example of case (i), i.e. the SCKV
\be
S = u^2 \partial_u + (r^2 / 2) \partial_v + u r \partial_r
\label{eq:october181137}
\ee
together with two non-special CKV given by equations (\ref{eq:thirtysix}) and (\ref{eq:october33A}) or (\ref{eq:october33B}) with $n=2$.

When $n=0$, i.e.,
\be
\tau(u) = \omega
\label{eq:october181157}
\ee
we have an example of case (iii), i.e., the KV
\be
X = \partial_u
\ee
together with two non-special CKV given by
\ba
C_1 & = & \sin (qu) \partial_u - \case{1}{4} q^2 r^2 \sin (qu) \partial_v + \case{1}{2} q r \cos (qu) \partial_r \, , \nonumber\\
C_2 & = & \cos (qu) \partial_u - \case{1}{4} q^2 r^2 \cos (qu) \partial_v - \case{1}{2} q r \sin (qu) \partial_r \, , \nonumber
\ea
if $\omega = q^2 / 8$, or
\ba
C_1 & = & \sinh (qu) \partial_u + \case{1}{4} q^2 r^2 \sinh (qu) \partial_v + \case{1}{2} q r \cosh (qu) \partial_r \, , \nonumber\\
C_2 & = & \cosh (qu) \partial_u + \case{1}{4} q^2 r^2 \cosh (qu) \partial_v + \case{1}{2} q r \sinh (qu) \partial_r \, , \nonumber
\ea
if $\omega = -q^2 / 8$.

When $n = 1$ equation (\ref{eq:thirdorder2}) yields
\be
\tau(u) = \omega u^{-2}
\label{eq:october181102}
\ee
where $\omega$ is a constant. There are three possibilities:
\begin{tabbing}

{00000000} \= {0000000000000000000000000000000000000000000000000000000000} \kill
$(a)$   \>  $\omega < 1/8 $, i.e., $8 \omega = 1-l^2$, $a(u) = u ( \beta + \zeta u^l + \sigma u^{-l})$. \\

$(b)$  \>  $\omega = 1/8 $, $a(u) = u ( \beta + \zeta \ln |u| + \sigma (\ln |u|)^2)$.  \\

$(c)$  \> $\omega > 1/8 $, i.e., $8 \omega = 1+l^2$, $a(u) = u ( \beta + \zeta \cos (l \ln |u|) + \sigma \sin (l \ln |u|))$.
\end{tabbing}
The quantities $l$, $\beta$, $\zeta$ and $\sigma$ are constants.
\begin{table}
\caption{The most general conformal symmetry classes with $L$, $M$, $N$ restricted to be functions of $u$ only, for $\tau(u) \ne 0$ and $\tau(u) = 0$. The last column gives the form of the function $F=H_{,yy} + H_{,zz}$ which appears in the Ricci tensor. Only class $Div$ can be vacuum (if $4 \delta(\theta) + \delta(\theta)_{, \theta \theta}=0$), all others are strictly non-vacuum solutions. Each of the spacetimes $A - D$ are conformally related to their counterpart spacetime $Aiv - Div$
respectively via the transformation (\ref{eq:conftransf}) given in section \ref{sec:typeNdual}.}
\footnotesize\rm
\begin{tabular*}{\textwidth}{@{}l*{15}{@{\extracolsep{0pt plus12pt}}l}}
\br
Class & Algebra & Orbits & Metric function $H$ & $F$\\

\mr

$A$ & ${\cal C}_5 \supset {\cal H}_2$ & $4$  & $\tau(u) r^2 + l e^{2 m \theta} r^{-2}$ & $4 \tau(u)+ 4 l (1+m^2) e^{2 m \theta} r^{-4}$ \\

$B$ & ${\cal C}_6 \supset {\cal G}_3$ & $4$  & $\tau(u) r^2 + l(\sigma z - \rho y)^{-2}$ & $4 \tau(u) + 6 \, l  (\sigma^2 + \rho^2) (\sigma z - \rho y)^{-4}$ \\

$C$ & ${\cal C}_5 \supset {\cal G}_2$ & $4$  & $\tau(u) r^2 + \delta r^{-2}$ & $4 \tau(u) + 4 \delta r^{-4}$\\

$D$ & ${\cal C}_4 \supset {\cal G}_1$ & $4$  & $\tau(u) r^2 + \delta(\theta) r^{-2}$ & $4 \tau(u) + (4 \delta(\theta) + \delta(\theta)_{, \theta \theta}) r^{-4}$\\

$Aiv$ & ${\cal S}_5 \supset {\cal H}_4$ & $4$  & $l e^{2 m \theta} r^{-2}$ & $4l e^{2 m \theta}(1 + m^2)r^{-4}$  \\

$Biv$ & ${\cal S}_6 \supset {\cal H}_5$ & $4$  & $l (\sigma z - \rho y)^{-2}$ & $6l (\sigma^2 + \rho^2) (\sigma z - \rho y)^{-4}$ \\

$Civ$ & ${\cal S}_5 \supset {\cal H}_4$ & $4$  & $\delta r^{-2}$ & $4 \delta r^{-4}$\\

$Div$ & ${\cal S}_4 \supset {\cal H}_3$ & $3$  & $\delta(\theta) r^{-2}$ & $(4 \delta(\theta) + \delta(\theta)_{, \theta \theta}) r^{-4}$\\

\br
\end{tabular*}
\label{tab:maxfunctionofu}
\end{table}

We now consider the conformal symmetries of classes $A$ - $D$. (See also table \ref{tab:maxfunctionofu}.)
\vspace{5mm}

{\bf Class A}.
$H = \tau(u) r^2 + l e^{2 m \theta} r^{-2}$. This is an example of an isometry class 1 solution.
Apart from the three CKV (\ref{eq:thirtysix}) and the KV $k$, the spacetime admits the HKV
\be
Z = 2 m v \partial_v + m r \partial_r  + 2 \partial_\theta \, ,
\label{eq:HKV2i}
\ee
and so the spacetime admits a ${\cal C}_5 \supset {\cal H}_2$.
As an example consider the case in which $\tau(u)$ is given by (\ref{eq:october181145}) ($n \ne 0,1,2$) for which the three CKV are all non-special.
The remaining Lie brackets are
\[
[ k, Y_i ] = [ Z, Y_i ] = 0 \, , \qquad [ k, Z] = 2 m Z \, .
\]

There are cases in which a second KV arises. First, if $\tau(u)$ is given by (\ref{eq:october181157}), there is the additional KV
$X_2 = \partial_u$, corresponding to isometry class 4 with $\epsilon = 0$, and the spacetime admits a ${\cal C}_5 \supset {\cal H}_3$.
Second, with (\ref{eq:october181102}) then we have the KV
\be
X_3 =  m (u \partial_u - v \partial_v) - \partial_\theta \, , \qquad [ k, X_3 ] = - m k \, .
\ee
The coordinate transformations $s = r \sin (\phi - \theta), t = r \cos (\phi - \theta)$ with $\phi = - ( \ln |u| )/m$ put the metric  function $H$ into the form
$H=u^{-2} [\omega (s^2 + t^2) + l (s^2 + t^2)^{-1} \exp(-2 m \arctan (s / t))]$, i.e., $H = u^{-2} W(s,t)$ so that the spacetime is of isometry class 3
and admits a ${\cal C}_5 \supset {\cal H}_3$.

{\bf Class B}.
$H = \tau(u) r^2 + l(\sigma \sin \theta - \rho \cos \theta)^{-2} r^{-2}$. This spacetime admits a ${\cal C}_6 \supset {\cal G}_3$
with basis consisting of the three CKV (\ref{eq:thirtysix}), the KV $k$ and the following two KVs.
If $f_I$, $I=1,2$, are two independent solutions of (\ref{eq:sixteen}) then the two KVs are
\be
\fl X_I = f'_I (\sigma \cos \theta + \rho \sin \theta) r \partial_v + f_I (\sigma \cos \theta + \rho \sin \theta) \partial_r
- r^{-1} f_I (\sigma \sin \theta - \rho \cos \theta) \partial_\theta \, ,
\label{eq:KVbasiscase2ii}
\ee
and we have the Lie brackets
\[
\fl [ X_1, X_2 ] = (\rho^2 + \sigma^2) \, m k \, , \qquad [ k, X_I ] = 0 \, , \qquad [ k, Y_i ] = 0 \, , \qquad [ X_I, Y_i ] = C^J_{Ii} \, X_J \, ,
\]
where $m=f_1 f'_2 - f_2 f'_1$ is a nonzero constant, $m = 0$ implies the two KV are not independent.
This is an example of an isometry class $1$ solution. However, it admits a ${\cal G}_3$ subalgebra which is distinct
from any of the ${\cal G}_3$ isometry classes of \cite{sippel86}. However, it is isomorphic to the class $1i$ algebra.

As an example consider the case (\ref{eq:october181116}) with $\omega = q^2 / 8$. The solutions for $a(u)$ are
\be
a_i(u) = u^2, u^2 \sin (q u^{-1}), u^2 \cos (q u^{-1})
\ee
and the solutions $f_I$, of equation (\ref{eq:thirdorder}) are
\[
f_1(u) = u \sin (q/ 2u) \, , \qquad f_2(u) = u \cos (q / 2u) \, .
\]
The CKV corresponding to $a_1(u)$ is the SCKV (\ref{eq:october181137}) while $a_2(u)$ and $a_3(u)$ correspond to non-special CKV.
The Lie brackets are
\[
\fl [ X_1, Y_1 ] = [ X_1, Y_3 ] = - [ X_2, Y_2 ] = \case{1}{2} q X_2 \, , \qquad
[ X_1, Y_2 ] = [ X_2, Y_3 ] = - [ X_2, Y_1 ] = \case{1}{2} q X_1 \, .
\]
Since this class of spacetime admits a ${\cal C}_6$ it must be conformally related to either $Biv$ or one of the plane wave spacetimes of
isometry class 10 in section \ref{sec:SandG}.

Note that if $\tau(u) = q$ is constant, one of the non-special CKV degenerates into the KV $X_3 = \partial_u$, so that the spacetime admits a
${\cal C}_6 \supset {\cal G}_4$. The Lie brackets are
\ba
{[}k, X_i ] = 0,  \qquad && [ X_1, X_2 ] = - (\rho^2 + \sigma^2) q k \, ,
\nonumber\\
{[} X_1, X_3 ] = - q X_2, \qquad && [ X_2, X_3 ] = q X_1 \, .
\nonumber
\ea

This is an example of a class 1 specialization admitting a ${\cal G}_3$ subalgebra which is distinct from any in \cite{sippel86}. However, this subalgebra
is isomorphic to isometry class $1i$.

{\bf Class C}.
$H = \tau(u) r^2 + \delta r^{-2}$, where $\delta$ is a constant. This spacetime admits a ${\cal C}_5 \supset {\cal G}_2$ and is an example of an isometry
class 2 solution. The three CKV are as in (\ref{eq:thirtysix}) and there is a ${\cal G}_2$ subalgebra given by the commuting KVs $k$
and $\partial_\theta$.

Consider the example (\ref{eq:october181102}). In each of the cases $(a)$, $(b)$ and $(c)$ the $\beta$ term corresponds to the HKV
\be
Z = 2 u \partial_u + r \partial_r \: ,
\ee
while the $\zeta$ and $\sigma$ terms correspond to non-special CKV (provided that in (a) $l \ne \pm 1$).
Thus the spacetime with $H = \delta r^{-2} + \omega u^{-2} r^2$ where $\delta$ and $\omega$ are non-zero constants,
admits a ${\cal C}_5$, a basis for which is provided by $k$, $X_2$, $Z$, non-special CKVs $C_1$, $C_2$ given by

\noindent \emph{Solution (a)}
\ba
C_1 & = & u^{l+1} \partial_u + \case{1}{4} (u^{l+1})'' r^2 \partial_v
+ \case{1}{2} (u^{l+1})' r \partial_r \: , \nonumber\\
C_2 & = & u^{-l+1} \partial_u + \case{1}{4} (u^{-l+1})'' r^2 \partial_v
+ \case{1}{2} (u^{-l+1})' r \partial_r \: , \nonumber
\ea
and Lie brackets
\ba
& & [ k, Z ] = [ k, C_1 ] = [ k, C_2 ] = 0 \, , \qquad [ X_2, Z ] = [ X_2, C_1 ] = [ X_2, C_2 ] = 0 \, , \nonumber\\
& & [ Z, C_1 ] = 2 l \, C_1  \, , \qquad [ Z, C_2 ] = -2 l \, C_2  \, , \qquad [ C_1, C_2 ] = - l \, Z  \, .  \nonumber
\ea

\noindent \emph{Solution (b)}
\ba
C_1 & = &  u \ln |u| \partial_u + \case{1}{4} ( u \ln |u|)'' r^2 \partial_v
+ \case{1}{2} ( u \ln |u|)' r \partial_r \: , \nonumber\\
C_2 & = & u (\ln |u|)^2 \partial_u + \case{1}{4} ( u (\ln |u|)^2)'' r^2 \partial_v
+ \case{1}{2} (u (\ln |u|)^2)' r \partial_r \: , \nonumber
\ea
and Lie brackets
\ba
& & [ k, Z ] = [ k, C_1 ] = [ k, C_2 ] = 0 \, , \qquad [ X_2, Z ] = [ X_2, C_1 ] = [ X_2, C_2 ] = 0 \, , \nonumber\\
& & [ Z, C_1 ] =  Z  \, , \qquad [ Z, C_2 ] = 2 C_1  \, , \qquad [ C_1, C_2 ] = C_2  \, .  \nonumber
\ea

\noindent \emph{Solution (c)}
\ba
C_1 & = &  u \cos \phi \partial_u + \case{1}{4} ( u \cos \phi)'' r^2 \partial_v
+ \case{1}{2} ( u \cos \phi)' r \partial_r \: , \nonumber\\
C_2 & = & u \sin \phi \partial_u + \case{1}{4} ( u \sin \phi)'' r^2 \partial_v
+ \case{1}{2} ( u \sin \phi)' r \partial_r \: , \nonumber
\ea
where $\phi = l \ln |u|$ and Lie brackets
\ba
& & [ k, Z ] = [ k, C_1 ] = [ k, C_2 ] = 0 \, , \qquad [ X_2, Z ] = [ X_2, C_1 ] = [ X_2, C_2 ] = 0 \, , \nonumber\\
& & [ Z, C_1 ] =  -2 l \, C_2  \, , \qquad [ Z, C_2 ] = 2 l \,  C_1  \, , \qquad 2 [ C_1, C_2 ] = l \, Z  \, .  \nonumber
\ea

Note that in the special case in which $\tau$ is constant in equation (\ref{eq:thirdorder2}), the spacetime admits a ${\cal C}_5 \supset {\cal G}_3$  and is
an example of an isometry class 6 solution which admits two proper CKV. This case will be discussed in the next section.

{\bf Class D}.
$H = \tau(u) r^2 + \delta(\theta) r^{-2}$, where $\delta(\theta)$ is an arbitrary function of $\theta$ which differs
from those appearing in classes $A$, $B$ and $C$. This spacetime admits only the KV $k$ and the three CKV given by
equation (\ref{eq:thirtysix}), one of which may degenerate into a SCKV depending on the function $\tau(u)$.
\vspace{1mm}

We now consider in detail the spacetimes of type $iv$ in Table \ref{tab:maxtypes}, i.e., those for which $\tau(u)=0$ so
that $H = \delta(\theta) r^{-2}$.
Spacetimes of this type only admit SCKV. In each of the four classes $Aiv-Div$ the spacetimes admit the
KV $k$ together with a second KV $X_2$, a proper HKV $Z$ and a proper SCKV $S$ given by
\be
X_2 = \partial_u \: , \qquad
Z = 2u \partial_u + r \partial_r \; , \qquad
S = u^2 \partial_u + {r^2 \over 2} \partial_v + u r \partial_r \; .
\ee
with Lie brackets
\[
\fl
[ k, X_2 ] = [ k, Z ] = [ k, S ] =  0 \, , \qquad
[ X_2 , Z ] = 2 X_2 \, , \qquad [ X_2, S ] = Z  \, , \qquad [ Z, S ] = 2 S \, .
\]
$X_2$, $Z$ and $S$ form a $so(2,1)$ subalgebra. For a class $Div$ spacetime, these are the only conformal symmetries
admitted, i.e., it admits a ${\cal S}_4 \supset {\cal H}_3 \supset {\cal G}_2$, but the other classes admit additional KVs.
Each of the spacetimes $A - D$ are conformally related to their counterpart spacetime $Aiv - Div$ respectively via the
transformation (\ref{eq:conftransf}) given in section \ref{sec:typeNdual}.
Further, each of the spacetimes $A - D$ will possess the $so(2,1)$ subalgebra structure formed by
$\{ X_2, Z, S \} \equiv \{ Y_i \}$, $i=1,2,3$.

{\bf Class Aiv}.
$H = l e^{2 m \theta} r^{-2}$. This is an example of an isometry class 7 solution. This spacetime admits the
${\cal S}_5 \supset {\cal H}_4 \supset {\cal G}_3$ consisting of  $k$, $X_2$, $Z$, $S$ and a third KV, $X_3$ given by
\[
X_3 = m \left(u\partial_u - v\partial_v \right) - \partial_\theta \: ,
\]
and the additional Lie brackets are
\[
\fl [X_3 , Z ] = 0 \, , \qquad [ X_2, X_3 ] = m X_2 \, , \qquad [ k, X_3 ] = - m k \, ,\qquad [ X_3, S ] = m S \, .
\]

{\bf Class Biv}.
$H = l (\sigma \sin \theta - \rho \cos \theta)^{-2} r^{-2} = l (\sigma z - \rho y)^{-2}$. Equation (\ref{eq:sixteen}) gives
\[
c_2(u) = \sigma(mu + n) \, , \qquad c_3(u) = \rho(mu + n) \, ,
\]
where $m$ and $n$ are constants.
The spacetime admits the ${\cal S}_6 \supset {\cal H}_5 \supset {\cal G}_4$ consisting of the $k$, $X_2$, $Z$ $S$,
and two additional KVs given by
\ba
\fl X_3  =  (\sigma \cos \theta + \rho \sin \theta ) \partial_r  - r^{-1} (\sigma \sin \theta - \rho \cos \theta ) \partial_\theta \: , \nonumber\\
\fl X_4  =  r (\sigma \cos \theta + \rho \sin \theta ) \partial_v +  u (\sigma \cos \theta + \rho \sin \theta ) \partial_r
- u r^{-1} (\sigma \sin \theta - \rho \cos \theta ) \partial_\theta \: .
\label{eq:basiscase1ii}
\ea
The additional Lie brackets are
\ba
\fl [ k, X_3 ] = [ k, X_4 ] = [X_2 , X_3 ] = [X_3 , Z ] = [X_4 , S ] = 0 \, , \nonumber\\
\fl [ X_2 , X_4 ] = X_3 \, , \qquad [ X_3, X_4 ] = (\rho^2 + \sigma^2) \, k \, , \qquad
[ X_3, S ] = X_4 \, , \qquad [ X_4, Z ] = - X_4  \, .  \nonumber
\ea
This spacetime is a class $8(\epsilon = 0)$ specialization and is conformally
related to the type $N$ with a ${\cal G}_6$ in equation (\ref{eq:defriseG6}).

{\bf Class Civ}.
$H = \delta r^{-2}$, where $\delta$ is a constant. This is an example of an isometry class  6 solution.
This spacetime admits the ${\cal S}_5 \supset {\cal H}_4 \supset {\cal G}_3$ consisting of $k$, $X_2$, $Z$, $S$ and the KV
$X_3 = \partial_\theta$. The additional Lie brackets are
\[
[ k, X_3 ] = [ X_2, X_3 ] = [ X_3, Z ] = [ X_3, S ] = 0\, .
\]

\section{Conformal symmetries for the type $N$ isometry classes}\label{sec:SandG}

We now consider each of the isometry classes of Sippel and Goenner \cite{sippel86} in order to determine whether
they admit any further conformal symmetries. We have obtained most of the following without specializing to
vacuum. However, we have found that specializing to vacuum only provides a few additional metrics, namely isometry class 2 specializations.
The results are summarized in tables \ref{tab:sippel} and \ref{tab:planewaves}.
\begin{table}
\caption{The conformal symmetry classes corresponding to the Sippel and Goenner isometry classes. Specializations which already appear in the previous tables are not listed here.}
\footnotesize\rm
\begin{tabular*}{\textwidth}{@{}l*{15}{@{\extracolsep{0pt plus12pt}}l}}
\br

Class & Algebra & Orbits & Metric function $H$ &  $F$\\ \hline

$1$  & ${\cal G}_1$ & $1n$ &  Arbitrary & $H_{,yy} + H_{,zz}$ \\

$1i$  & ${\cal G}_3$ & $3n$ & $H(u,z)$  & $H_{,zz}$ \\

$2$  & ${\cal G}_2$ & $2n$ & $H(u,r)$ & $H_{,rr} +r^{-1} H_{,r}$ \\

$2i$  & ${\cal H}_3$ & $3t$ & $K (\alpha u + \beta)^q \ln |r|$ & $0$ \\

$2ii$  & ${\cal H}_3$ & $3t$ & $K e^{[- \Theta u / \beta ]} \ln |r|$ & $0$ \\

$2iii$  & ${\cal S}_3 \supset {\cal G}_2$ & $3t$ & $e^{g(u)} \ln |r|$, $g(u)$ given by (\ref{eq:class2functiong}) & $0$ \\

$2iv$  & ${\cal C}_3 \supset {\cal G}_2$ & $3t$ & $H$ satisfies $(\ref{eq:Hsatisfycase2})$ with $a'''(u) \ne 0$ & $H_{,rr} +r^{-1} H_{,r}$ \\

$3$  & ${\cal G}_2$ & $2t$ & $u^{-2}W(s,t)$ & $u^{-2} (W_{,ss} + W_{,tt})$ \\

$4$  & ${\cal G}_2$ & $2t$ & $W(s,t)$ & $W_{,ss} + W_{,tt}$ \\

$5$  & ${\cal G}_3$ & $3t$ & $u^{-2}W(r)$ & $ u^{-2} (W_{,rr} +r^{-1} W_{,r})$\\

$5i$ & ${\cal S}_4 \supset {\cal G}_3$  & $4$ & $u^{-2} \zeta \ln |r|$ & $0$ \\

$5ii$ & ${\cal C}_4 \supset {\cal G}_3$  & $4$ & $u^{-2} (\delta r^{-\sigma} - \sigma (2 - \sigma)^{-2} r^2)$ &
$u^{-2}(\delta \sigma^2 r^{-(\sigma + 2)} - 4 \sigma (2 - \sigma)^{-2})$ \\

$6$  & ${\cal G}_3$ & $3t$ & $W(r)$ & $W_{,rr} +r^{-1} W_{,r}$ \\

$6i$ & ${\cal C}_5 \supset {\cal G}_3$ & $4$ & $\nu r^2 / 4 + \delta r^{-2}$ & $\nu + 4 \delta r^{-4}$\\

$6ii$ & ${\cal S}_5 \supset {\cal G}_3$ & $4$ & $\delta r^{-2}$ & $4 \delta r^{-4}$\\

$6iii$ & ${\cal H}_4$ & $4$ & $\zeta \ln |r|$ & $0$\\

$6iv$ & ${\cal H}_4$ & $4$ & $\delta r^{-\sigma}$ & $\delta \sigma^2 r^{-(\sigma + 2)}$\\

$7$  & ${\cal G}_3$ & $3t$ & $e^{2 c \theta} \, W(r)$ & $e^{2 c \theta}(W_{,rr} + r^{-1} W_{,r} + 4 c^2 r^{-2} W)$ \\

$7i$ & ${\cal H}_4$ & $4$ & $\delta r^{-\sigma} e^{2 c \theta}$, $\sigma \ne 2$ &
$\delta e^{2 c \theta} (\sigma^2 + 4 c^2) r^{-(\sigma + 2)}$\\

$7ii$ & ${\cal S}_5 \supset {\cal H}_4$ & $4$ & $\delta r^{-2} e^{2 c \theta}$ & $4 \delta e^{2 c \theta} (1 + c^2) r^{-4}$\\

$8$  & ${\cal G}_3$ & $3t$ & $e^{2t} \, W(s)$ & $[(\eta^2 + \sigma^2) W_{,ss} + 4 \epsilon^2 (\eta^2 + \sigma^2)^{-1}W ] e^{2t}$ \\

$8(\epsilon = 0)$  & ${\cal G}_4$ & $4$ & $W(s)$ & $(\eta^2 + \sigma^2) W_{,ss}$ \\

$8(\epsilon = 0)i$  & ${\cal H}_5$ & $4$ & $Ks^c, c \ne -2$ & $(\eta^2 + \sigma^2) \, c(c-1) Ks^{c-2}$ \\

$9$  & ${\cal G}_5$ & $4$ & $K e^{2(\eta y - \sigma z)}$ & $4 (\eta^2 + \sigma^2)H$ \\

\br
\end{tabular*}
\label{tab:sippel}
\end{table}
\subsection{Isometry class 1}
This is the general case which admits only the KV $k$. However, there are specializations admitting additional conformal symmetries,
e.g., classes $A$ and $D$ belong to isometry class 1.
Further, in section \ref{sec:typeNdual} we present examples of this class admitting ${\cal H}_2$, ${\cal S}_2$ and ${\cal S}_3$.
In the vacuum case there can be at most an ${\cal S}_3$, which has at most an ${\cal H}_2$ subalgebra.

\subsection{Isometry class 1$i$}

In this class $H$ is independent of one of the spatial coordinates, i.e., $H = H(u,z)$.
The condition (\ref{eq:Hsatisfy}), and exclusion of type $O$, gives $\gamma = c''_y(u) = a'''(u) = 0$ and
\ba
H_{,u} (\rho u^2 + \alpha u + \beta) + H_{,z} [(2 \rho u + \alpha - \Theta) z / 2 + c_z(u)]
\nonumber\\
+ c''_z(u) z + E'(u) + (2 \rho u + \alpha + \Theta) H = 0 \, ,
\nonumber
\ea
where $\rho$, $\alpha$ and $\beta$ are arbitrary constants. If $H(u, z)$ is an arbitrary function then it follows immediately from the above that
the only CKV which can occur are $k$ and
\be
X_2 = \partial_y \, , \qquad X_3 = y \partial_v + u \partial_y \, .
\ee
Thus we have a ${\cal G}_3$ with Lie brackets
\[
[ k, X_2 ] = [ k, X_3 ] = 0 \, , \qquad [ X_2, X_3 ] = k \, .
\]

\subsection{Isometry class 2}
In this case the function $H$ has the form $H = H(u,r)$.
The condition (\ref{eq:Hsatisfy}), and exclusion of type $O$, gives
$c_A(u)=0$ and
\be
\fl H_{,u} \, a(u) + r H_{,r} (a'(u) - \Theta) / 2 + a'''(u) r^2 /4 + E'(u)
= - (a'(u) + \Theta) H \: .
\label{eq:Hsatisfycase2}
\ee
In this case the CKV (\ref{eq:refinedCKVpolars}) has the form
\be
Y = a(u) \partial_u + (-\Theta v + a''(u) r^2 / 4 + E(u)) \partial_v
+ \psi r \partial_r - \gamma \partial_\theta \: .
\label{eq:ckvcase2}
\ee
If $H$ is an arbitrary function then it follows immediately from (\ref{eq:Hsatisfycase2}) that the only CKV which occur
are the two KV already listed in \cite{sippel86}, i.e., $k$ and
\be
X_2 = \partial_\theta \: , \qquad [ k , X_2 ] = 0 \: .
\label{eq:KVsforcase2}
\ee
Since $X_2$ is a KV we may put $\gamma=0$.

The partial differential equation is too difficult to solve in general. This class of pp-wave spacetime has axial symmetry and Barnes \cite{barnes01} has
determined restrictions on the structure constants of two-, three- and four-dimensional Lie algebras possessing a
one-dimensional cyclic Lie subalgebra. However, we have seen that there are specializations, given in class $C$ of section \ref{sec:typeN}, which admit
two or three proper CKV and further specializations admitting one non-special CKV, one SCKV, or one HKV are found in section \ref{sec:typeNdual}.

Solutions of the form $H = m(u) \, r^p + n(u) \, r^q$
can be found, where $p \ne q$ are constants. Substituting into equation (\ref{eq:Hsatisfycase2}) gives
\ba
(m'(u) r^p + n'(u) r^q) a(u) + (p m(u) r^p + q n(u) r^q) (a'(u) - \Theta)/2 \nonumber\\
+ a'''(u) r^2 /4 + E'(u) + (a'(u) + \Theta)(m(u) r^p + n(u) r^q) = 0 \: . \nonumber
\ea
One of $p,q$ must equal $2$ in order to avoid $a'''(u) = 0$; choose $q=2$ and assume $p \ne 0$ and $p \ne 2$ (otherwise type $O$). Then
equating coefficients of powers of $r$ we have  $E'(u) = 0$ and
\ba
& & a(u) m'(u) + (1 + p/2) m(u) a'(u) + (1 - p/2) m(u) \Theta = 0 \: , \label{eq:briansA1}\\
& & a'''(u)/4 + 2 n(u) a'(u) + n'(u) a(u) = 0 \: . \label{eq:briansA2}
\ea
If $\Theta=0$ and $p=-2$ then (\ref{eq:briansA1}) implies $m$ is constant and equation (\ref{eq:briansA2}) then yields three independent
solutions for $a(u)$ all of which may be proper CKV. This is class $C$ considered earlier. However, for arbitrary values of
$p$ and $\Theta$ and for fixed non-zero functions $m(u)$, $n(u)$, equation (\ref{eq:briansA1}) shows that there exists at most one
solution for $a(u)$ which satisfies both equations. If such an $a(u)$ exists there will be one additional conformal symmetry only, which
will be a non-special proper CKV if $a'''(u) \ne 0$. By eliminating $a(u)$ from equations (\ref{eq:briansA1}) and (\ref{eq:briansA2}) we arrive at a
difficult differential equation connecting $m(u)$ and $n(u)$ which gives the condition to be satisfied in order for the
spacetime to admit the additional conformal symmetry. As an example consider the case in which $n(u) = (1 - l^2)/8u^2$ where $l$ is
a constant. From solution $C$, case (a), equation (\ref{eq:briansA2}) leads to three independent solutions for
$a(u)$, namely $a(u)=u$, $a(u)= u^{1+l}$, $A(u) = u^{1-l}$. We consider these in turn:

$a(u)=u$, $m(u) = \sigma u^{[(p - 2) \Theta - (p+2)]/2}$, $\sigma$ is a constant and HKV
\[
Z = u \partial_u - \Theta v \partial_v + \psi r \partial_r \, , \qquad \psi = (1 - \Theta)/2 \, .
\]

$a(u)= u^{1+l}$, $m(u) = \sigma u^{-(p+2)(1+l)/2} \exp[(2 - p) \Theta (1+l)^{-1} u^{-l} / 2 ]$, $\sigma$ is a constant and proper CKV
\[
\fl C = u^{1+l} \partial_u +(l(1+l) u^{l-1} r^2 /4 - \Theta v) \partial_v + \psi r \partial_r \, ,
\qquad \psi = [(1 + l)u^{l} - \Theta]/2  \, .
\]

$a(u) = u^{1-l}$, $m(u) = \sigma u^{-(p+2)(1-l)/2} \exp[(2 - p) \Theta (1-l)^{-1} u^l / 2 ]$, $\sigma$ is a constant and proper CKV
\[
\fl C = u^{1-l} \partial_u +(l(l-1) u^{-(l+1)} r^2 /4 - \Theta v) \partial_v + \psi r \partial_r \, ,
\qquad \psi = [(1 - l)u^{-l} - \Theta]/2  \, .
\]
In all three cases the Lie brackets involving the extra CKV Y are
\[
[ k, Y ] = -\Theta k \, , \qquad [X_2, Y] = 0 \, .
\]

Let us consider the specialization to vacuum, in which case the metric function can be put in the form $H = e^{g(u)} \ln |r|$.
Only SCKV can occur in a vacuum spacetime. In general there will be no further conformal symmetry. However, we can obtain
SCKV of the form (\ref{eq:generalSCKV}) with the choice
\be
g'(u) = - ( \Theta + \alpha + 2 \rho u) (\rho u^2 + \alpha u + \beta)^{-1} \, , \qquad E'(u) = - \psi \, e^{g(u)} \, .
\ee
The function $g(u)$ is given by
\be
g(u) = - \ln |\rho u^2 + \alpha u + \beta| - \Theta \int (\rho u^2 + \alpha u + \beta)^{-1} du \, .
\label{eq:class2functiong}
\ee
In general there will be only one extra such (not necessarily proper) SCKV.
There are only two possibilities with an additional KV, $H=K \ln |r|$ and $H = K(\Theta u + \beta)^{-2} \ln |r|$ (the constant $\beta$ can be transformed away),
and in both cases they admit further CKV, i.e.,
${\cal H}_4 \supset {\cal G}_3$ (isometry class 6 specialization) and ${\cal S}_4 \supset {\cal G}_3$ (isometry class 5 specialization) respectively, and these
spacetimes are conformally related to one another, see section \ref{sec:typeNdual}. These arise because some of the parameters $\rho$, $\alpha$, $\beta$ and
$\Theta$ can be chosen in such a way as they do not appear in the function $g(u)$ and so are arbitrary constants.

There are three possibilities for proper HKV, $\rho = 0$, $\alpha \ne \Theta$ and
\newline
(i) ${\cal H}_3 \supset {\cal G}_2$, $\alpha = 0$, $\beta \ne 0$, $H = K e^{[- \Theta u / \beta ]} \ln |r|$, with HKV of the form
\[
\fl Z = \beta \partial_u - (\Theta v + \beta K e^{[- \Theta u / \beta ]} / 2) \, \partial_v - (\Theta / 2) \, r \partial_r \, , \qquad
[ Z, k ] = \Theta k \, , \qquad [ Z, X_2 ] = 0 \, .
\]
(ii) ${\cal H}_3 \supset {\cal G}_2$, $\alpha \ne 0$, $\beta$ arbitrary, $\alpha + \Theta \ne 0$, $H=K (\alpha u + \beta)^q \ln |r|$, $q = -(\alpha + \Theta) / \alpha$.
For the case with $\Theta = 0$ the HKV is of the form
\[
\fl Z = (\alpha u + \beta) \partial_u - (K / 2) \ln|(\alpha u + \beta)| \, \partial_v - (\alpha / 2) \, r \partial_r \,  , \qquad
[ Z, k ] = 0 \, , \qquad [ Z, X_2 ] = 0 \, ,
\]
and for the case with $\Theta \ne 0$ the HKV is of the form
\ba
\fl Z  =  (\alpha u + \beta) \partial_u +
(- \Theta v + K (\alpha / 2 \Theta)(\alpha - \Theta) (\alpha u + \beta)^{-\Theta / \alpha}) \, \partial_v
- ((\alpha - \Theta)/2) \, r \partial_r \,  , \nonumber\\
 {[ Z, k ]} = \Theta k \, , \qquad {[ Z, X_2 ]} = 0 \, .
\ea
(iii) ${\cal H}_4 \supset {\cal G}_3$, $\alpha \ne 0$, $\beta$ arbitrary, $\alpha + \Theta = 0$, $H=K \ln |r|$, discussed above.

Regarding proper SCKV, if one makes the choice $\alpha = \Theta$ and $\beta = 0$ then $H = K u^{-2} \ln |r|$ and we have the ${\cal S}_4 \supset {\cal G}_3$ discussed above.
Otherwise there is only ${\cal S}_3 \supset {\cal G}_2$.
There are no ${\cal S}_4 \supset {\cal H}_3 \supset {\cal G}_2$ since they would necessarily be specializations of
the ${\cal H}_3 \supset {\cal G}_2$ which admit no further CKV.

\subsection{Isometry class 3}
In this case the function $H$ has the form $H = u^{-2} W(s,t)$,
where $s = y \sin \phi - z \cos \phi$, $t = y \cos \phi + z \sin \phi$, $\phi = \epsilon \ln |u|$ and $\epsilon$ is an arbitrary constant.
We shall relabel $c_y(u)$, $c_z(u)$ as $c_2(u)$, $c_3(u)$, respectively.
The general CKV takes the form
\ba
Y^u & = & a(u) \: , \nonumber\\
Y^v & = & -\Theta v + a''(u) \, (s^2 + t^2)/4 +
c'_2(u) (s \sin \phi + t \cos \phi)
\nonumber\\
& & + c'_3(u) (t \sin \phi - s \cos \phi)
+ E(u) \: , \nonumber\\
Y^s & = & \psi s + a(u) \phi' t + \gamma t + c_2(u) \sin \phi - c_3(u) \cos \phi \: , \nonumber\\
Y^t & = & \psi t - a(u) \phi' s - \gamma s + c_2(u) \cos \phi + c_3(u) \sin \phi \: .
\label{eq:ckvcase3}
\ea
The condition (\ref{eq:Hsatisfy}) takes the form
\ba
& & H_{,s} [a(u) \phi' t + \psi s + \gamma t + c_2(u) \sin \phi - c_3(u) \cos \phi] \nonumber\\
& &  + H_{,t} [-a(u) \phi' s + \psi t - \gamma s + c_2(u) \cos \phi + c_3(u) \sin \phi] \nonumber\\
& &  + c''_2(u) (s \sin \phi + t \cos \phi) + c''_3(u) (t \sin \phi - s \cos \phi)  \nonumber\\
& &  + a'''(u) (s^2 + t^2) / 4  +  H' \, a(u) + (a'(u) + \Theta) H + E'(u) = 0 \: .
\label{eq:class3condition}
\ea
If $W$ is an arbitrary function then it follows that the only CKV which occur are the two KV already listed in \cite{sippel86}, i.e., $k$ and
\be
X_2 =   \epsilon (y \partial_z - z \partial_y)  +  (u \partial_u - v \partial_v) \, , \qquad [ k , X_2 ] = - k \: .
\label{eq:KVsforcase3}
\ee
There are at least two class 3 spacetimes which admit further conformal symmetry, class $Aii$ with ${\cal C}_5 \supset {\cal H}_3$ and
dual spacetime $D7i$ with ${\cal S}_4 \supset {\cal H}_3$, see section \ref{sec:typeNdual}.

\subsection{Isometry class 4}
In this case the function $H$ has the form $H = W(s,t)$,
where $\phi= \epsilon u$, $\epsilon$ is an arbitrary constant and the general CKV takes the form (\ref{eq:ckvcase3}).
The condition (\ref{eq:Hsatisfy}) takes the form (\ref{eq:class3condition}).
If $W$ is an arbitrary function then it follows that the only CKV which occur are the two KV already listed in \cite{sippel86}, i.e.,
$k$ and
\be
X_2 =  \epsilon (y \partial_z - z \partial_y) +  \partial_u \: , \qquad [ k , X_2 ] = 0 \: .
\label{eq:KVsforcase4}
\ee
There is at least one class 4 spacetime which admits further conformal symmetry, class $Aiii$ with ${\cal C}_5 \supset {\cal H}_3$, which
corresponds to $\epsilon = 0$, i.e., $W(s,t) = W(-z, y)$ and $X_2 = \partial_u$.

\subsection{Isometry class 5}
In this case the function $H$ has the form $H = u^{-2} W(r)$.
The condition (\ref{eq:Hsatisfy}), and exclusion of type $O$, gives $c_A(u)=0$ and
\ba
-2u^{-3} W a(u) + r W_{,r} u^{-2} (a'(u) - \Theta) /2 + a'''(u)r^2 /4 + E'(u)
\nonumber\\
= -(a'(u) + \Theta) W u^{-2} \: ,
\label{eq:Hsatisfycase5}
\ea
and the general CKV takes the form (\ref{eq:ckvcase2}). If $W$ is an arbitrary function then it follows immediately
from (\ref{eq:Hsatisfycase5}) that the only CKV which occur are the three KV already listed in \cite{sippel86}, i.e., $k$ and
\be
X_2 = u\partial_u - v\partial_v  \: , \qquad
X_3 = \partial_\theta \: ,
\label{eq:KVsforcase5}
\ee
with Lie brackets
\[
[ k , X_2 ] = -k  \: , \qquad [ k , X_3 ] = 0 \: , \qquad [ X_2, X_3 ] = 0 \: .
\]

Let us now consider restrictions on the functional form of $W(r)$ which will give rise to additional CKV.
The differential equation (\ref{eq:Hsatisfycase5}) can be written in the form
\be
r W(r)_{,r} + f(u) W(r) = g(u) + h(u) r^2 \: .
\label{eq:Hsatisfycase567}
\ee
Differentiating (\ref{eq:Hsatisfycase567}) with respect to $u$ gives
\be
f'(u) W(r) = g'(u) + h'(u) r^2
\label{eq:Hsatisfycase567deru}
\ee
which, if $f'(u) \ne 0$, implies that $W_{,yy} = W_{,zz}$, $W_{,yz}=0$, i.e., type $O$. Hence, we must have $f$, $g$ and $h$ constant
to avoid type $O$. Put $f=\sigma$, $g=\zeta$, $h=\nu$, then equation (\ref{eq:Hsatisfycase567}) becomes
\[
(r^\sigma W)_{, r} = \zeta \, r^{\sigma - 1} + \nu \, r^{\sigma + 1}
\]
which integrates to give
\ba
\sigma \ne -2, 0 : \qquad & W(r) = \zeta/\sigma + \nu r^2 /(\sigma + 2) + \delta r^{-\sigma}
\nonumber\\
\sigma = -2 : \qquad & W(r) = - (\zeta/ 2) + \nu r^2 \ln |r| + \delta r^2
\nonumber\\
\sigma = 0 : \qquad & W(r) = \zeta \ln |r| + \nu r^2/2 + \delta
\nonumber
\ea
where $\delta$ is a constant. Each of the expressions for $W(r)$ contains an additive constant, $q$, which can be
set to zero with a transformation of the form $v \mapsto v -q / u$. The above equations can then be replaced by
\ba
\sigma \ne -2, 0 : \qquad & W(r) = \nu r^2 /(\sigma + 2) + \delta r^{-\sigma}
\label{eq:Wsolncase567}
\\
\sigma = -2 :  & W(r) = \nu r^2 \ln |r| + \delta r^2
\label{eq:Wsolncase567sig-2} \\
\sigma = 0 : & W(r) = \zeta \ln |r| + \nu r^2/2
\label{eq:Wsolncase567sigzero}
\ea
The differential equation (\ref{eq:Hsatisfycase567}) gives
\ba
& & a'(u) -2 u^{-1} a(u) + \Theta = \sigma (a'(u) - \Theta)/2  \, , \label{eq:case5a} \\
& & E'(u) u^2 = - \zeta (a'(u) - \Theta)/2 \, , \label{eq:case5b} \\
& & a'''(u) = -2 u^{-2} \nu (a'(u) - \Theta) \, . \label{eq:case5c}
\ea
where in equation (\ref{eq:case5b}), the rhs is zero except in the case $\sigma = 0$.
Exclusion of type $O$ (i.e., $W(r) \propto r^2$) requires $\delta \ne 0$ in equation (\ref{eq:Wsolncase567}), $\nu \ne 0$ in equation
(\ref{eq:Wsolncase567sig-2}) and $\zeta \ne 0$ in equation (\ref{eq:Wsolncase567sigzero}).

If $\sigma=2$ then equations (\ref{eq:case5a}) and (\ref{eq:case5b}) imply that the only conformal symmetries are the three KVs $k$, $X_2$ and $X_3$.
If $\sigma = -2$ then (\ref{eq:case5a}) gives
\be
a(u) = l u,
\ee
where $l$ is a constant, and equation (\ref{eq:case5c}) gives $a'''(u) = 0 = -2 u^{-2} \nu (l - \Theta)$ and since $\nu \ne 0$, we have $a(u) = \Theta u$
and again the only conformal symmetries are the KVs above.
If $\sigma = 0$ then (\ref{eq:case5a}) integrates to give
\be
a(u) = \Theta u + l u^2,
\label{eq:5sigmazero}
\ee
where $l$ is a constant and equation (\ref{eq:case5b}) gives $E(u) = - \zeta l \ln |u| + m$ where $m$ is a constant. From equations
(\ref{eq:case5c}) and (\ref{eq:5sigmazero}) we find $a'''(u) = 0 = -4 l \nu / u$, i.e., $l \nu = 0$. If $l = 0$ then only the three KV above exist, but if
$l \ne 0$ and $\nu = 0$, i.e., $W(r) = \zeta \ln |r|$, then in addition to the three KVs we have the proper SCKV
\be
S = u^2 \partial_u + ( - \zeta \ln |u| + r^2/2)\partial_v
+ u \, r\partial_r \: ,
\label{eq:case5SCKV1}
\ee
with $\psi = u$. The Lie brackets for the ${\cal S}_4$ are
\[
[ k, S ] = [ X_3, S ] =0 \; , \qquad [ X_2, S ] = S - \zeta k \; .
\]

If $\sigma \ne -2, 0, 2$ then (\ref{eq:case5a}) integrates to give
\be
a(u) = \Theta u + l \, u^{4 / (2-\sigma)},
\label{eq:case5solnfora}
\ee
where $l$ is a constant. Substituting this into equation (\ref{eq:case5c}) gives $\nu = -(2 + \sigma) \sigma / (\sigma - 2)^2$, i.e.,
\be
W(r) = -\sigma (2 - \sigma)^{-2} r^2 + \delta r^{-\sigma} \, .
\ee
This spacetime admits a ${\cal C}_4$ consisting of $k$, $X_2$, $X_3$ and the proper CKV
\be
C = u^q \left(
\partial_u + (2 + \sigma)(2 - \sigma)^{-2} r^2 u^{-2} \partial_v
+ (q r / 2 u) \partial_r \right) \: ,
\label{eq:case5CKV1}
\ee
where $q= 4 / (2 -\sigma)$. The Lie brackets are
\[
[ k, C ] = [ X_3, C ] = 0 \; , \qquad [ X_2, C ] = (q - 1) C \; .
\]

In summary, the isometry class 5 spacetime of Petrov type $N$ with metric
\be
ds^2 = - 2 du dv - 2 \zeta u^{-2} \ln |r| \, du^2 + dy^2 + dz^2
\label{eq:summarymetric5i}
\ee
admits an  ${\cal S}_4 \supset  {\cal G}_3$. The spacetimes with metric
\be
ds^2 =  -2 du dv - 2 u^{-2} [\delta r^{- \sigma} - \sigma (2- \sigma)^{-2} r^2]  \, du^2 + dy^2 + dz^2,
\label{eq:summarymetric5ii}
\ee
where $\sigma \ne -2, 0, 2$ admit an ${\cal C}_4 \supset  {\cal G}_3$. All other isometry class 5 type $N$ spacetimes admit only the
${\cal G}_3$.

\subsection{Isometry class 6}
In this case the function $H$ has the form $H = W(r)$.
The condition (\ref{eq:Hsatisfy}), and exclusion of type $O$, again gives $c_A(u)=0$ and
\be
r W_{,r} (a'(u) - \Theta) /2 + a'''(u) r^2 /4 + E'(u) = -(a'(u) + \Theta) W \: ,
\label{eq:Hsatisfycase6}
\ee
and the general CKV takes the form (\ref{eq:ckvcase2}). If $W$ is an arbitrary function then it follows immediately from
(\ref{eq:Hsatisfycase6}) that the only CKV which occur are the three KV already listed in \cite{sippel86}, i.e., $k$ and
\be
X_2 = \partial_u \: , \qquad
X_3 = \partial_\theta \: ,
\label{eq:KVsforcase6}
\ee
with Lie brackets
\[
[ k , X_2 ] = 0  \: , \qquad [ k , X_3 ] = 0 \: , \qquad [ X_2, X_3 ] = 0 \: .
\]
The differential equation (\ref{eq:Hsatisfycase6}), written in the form (\ref{eq:Hsatisfycase567}), gives equations identical to (\ref{eq:Wsolncase567}) to
(\ref{eq:Wsolncase567sigzero}) and leads to
\ba
& & a'(u) + \Theta = \sigma (a'(u) - \Theta)/2  \, , \label{eq:case6a} \\
& & E'(u) = -\zeta (a'(u) - \Theta)/2 \, , \label{eq:case6b} \\
& & a'''(u) = -2 \nu (a'(u) - \Theta) \, , \label{eq:case6c}
\ea
where, as previously, $\zeta = 0$ in equation (\ref{eq:case6b}) except in the case where $\sigma = 0$.

If $\sigma = 2$, equations (\ref{eq:case6a}) - (\ref{eq:case6c}) imply that $\Theta = 0$, $E'(u) = 0$,
\be
a'''(u) + 2 \nu a'(u) = 0 \, ,
\ee
and $H(r) = W(r) = \nu r^2 / 4 + \delta r^{-2}$. This is the special case of the class $C$ solution of section \ref{sec:typeN} in which
$\tau(u) = \nu /4$ is constant. If $\nu > 0$, put $\nu = l^2/2$
and the differential equation integrates to give
\be
a(u) = a_0 + a_1 \cos lu + a_2 \sin lu \: ,
\ee
where $a_1$, $a_2$ and $a_0$ are constants.
In this case we have a ${\cal C}_5$, consisting of $k, X_2, X_3$ and the two proper CKVs $C_1$, $C_2$
\ba
C_1 & = & \cos lu \partial_u
- (l^2 r^2 / 4) \cos lu  \partial_v - l {r \over 2} \sin lu  \partial_r \:, \label{eq:case6CKV1}\\
C_2 & = & \sin lu \partial_u
- (l^2 r^2 / 4) \sin lu  \partial_v + l {r \over 2} \cos lu  \partial_r \:,
\label{eq:case6CKV2}
\ea
with Lie brackets
\ba
& & [ k, C_1 ] = [ k, C_2 ] = [ X_3, C_1 ] = [ X_3, C_2 ] = 0 \, , \nonumber\\
& & [ C_1, C_2 ] = l X_2 \; , \qquad [ X_2 , C_1 ] = -l \, C_2 \; , \qquad [ X_2, C_2 ] = l \, C_1 \: .
\ea
For $\nu < 0$, we put $\nu = -l^2/2$ and the differential equation integrates to give
\be
a(u) = a_0 + a_1 \cosh lu + a_2 \sinh lu \: ,
\ee
where $a_1$, $a_2$ and $a_0$ are constants.
In this case we have a ${\cal C}_5$, consisting of $k, X_2, X_3$ and the two proper CKV $C_1$, $C_2$
\ba
C_1 & = & \cosh lu \partial_u
+ (l^2 r^2 / 4) \cosh lu  \partial_v + l {r \over 2} \sinh lu  \partial_r \:, \label{eq:case6CKV3}\\
C_2 & = & \sinh lu \partial_u
+ (l^2 r^2 / 4) \sinh lu  \partial_v + l {r \over 2} \cosh lu  \partial_r \:, \label{eq:case6CKV4}
\ea
with Lie brackets
\ba
& & [ k, C_1 ] = [ k, C_2 ] = [ X_3, C_1 ] = [ X_3, C_2 ] = 0 \, , \nonumber\\
& & [ C_1, C_2 ] = l X_2  \; , \qquad [ X_2 , C_1 ] = l \, C_2 \; , \qquad [ X_2, C_2 ] = l \, C_1 \: .
\ea
If $\nu =0$, i.e., $W(r) = \delta r^{- 2}$, the spacetime is that of class $Civ$ of section \ref{sec:typeN} which admits the
${\cal S}_5 \supset {\cal H}_4 \supset {\cal G}_3$.

If $\sigma = -2$, equation (\ref{eq:case6a}) gives $a'(u) = \Theta = 0$ and so only KV are possible.
If $\sigma = 0$ equation (\ref{eq:case6a}) gives
\be
a(u) = - \Theta u + l
\ee
where $l$ is a constant. Equation (\ref{eq:case6c}) becomes $a'''(u) = 0 = 4 \nu \Theta$. If $\Theta = 0$ then only the three KV exist.
However, if $\nu = 0$ and $\Theta \ne 0$ then $W(r) = \zeta \ln |r|$, in which case the spacetime admits an ${\cal H}_4 \supset {\cal G}_3$ consisting
of $k$, $X_2$ and $X_3$ and the proper HKV
\be
Z = u \partial_u + (v - \zeta u) \partial_v + r \partial_r \: ,
\label{eq:case6HKV}
\ee
with $\psi = 1$ and Lie brackets
\[
[ k , Z ] = k  \: , \qquad [ X_2, Z ] =   X_2 - \zeta \, k \: , \qquad [ X_3, Z] = 0 \: .
\]

If $\sigma \ne -2, 0, 2$ then $W(r)$ is given by equation (\ref{eq:Wsolncase567}) and equation (\ref{eq:case6a}) integrates to give
\be
a(u) = (\sigma + 2) (\sigma - 2)^{-1} \Theta u + l \, ,
\ee
where $l$ is a constant. Equation (\ref{eq:case6c}) then leads to $\nu (a'(u) - \Theta) = 0$.
The case with $\nu \ne 0$ and $a'(u) = \Theta$ admits only the three
KVs and $\nu = 0$ leads to an ${\cal H}_4 \supset {\cal G}_3$ consisting of the three KVs and the proper HKV
\be
2 Z =  (\sigma + 2) u \partial_u + (2 - \sigma) v \partial_v + 2 r  \partial_r  \: ,
\label{eq:case6HKVlast}
\ee
with $\psi = 1$ and Lie brackets
\[
2 [ k , Z ] = (2 - \sigma) k  \: , \qquad 2 [ X_2, Z ] =  (\sigma + 2) X_2 \: , \qquad [ X_3, Z] = 0 \: .
\]

In summary, the isometry class 6 spacetimes of Petrov type $N$ with metric
\be
ds^2 = - 2 du dv - 2 (\nu r^2 / 4 + \delta r^{-2}) du^2 + dy^2 + dz^2
\label{eq:summarymetric6i}
\ee
with $\nu \ne 0$ admits a ${\cal C}_5 \supset {\cal G}_3$. If $\nu = 0$ the metric
\be
ds^2 = - 2 du dv - 2 \delta r^{-2} du^2 + dy^2 + dz^2
\label{eq:summarymetric6ii}
\ee
admits a ${\cal S}_5 \supset {\cal H}_4 \supset {\cal G}_3$. The spacetime with metric
\be
ds^2 = - 2 du dv - 2 \zeta \ln |r| \, du^2 + dy^2 + dz^2
\label{eq:summarymetric6iii}
\ee
admits an ${\cal H}_4 \supset {\cal G}_3$ and the class of spacetimes with metric
\be
ds^2 = - 2 du dv - 2 \delta r^{- \sigma} du^2 + dy^2 + dz^2
\label{eq:summarymetric6iv}
\ee
with $\sigma \ne -2, 0, 2$ admits an ${\cal H}_4 \supset {\cal G}_3$. All other isometry class 6 spacetimes admit only the ${\cal G}_3$.

\subsection{Isometry class 7}
In this case the function $H$ has the form $H = W(r) \exp (2 c \, \theta)$,
where $y = r \cos \theta$, $z = r \sin \theta$ and $c \ne 0$ is a constant.
Note that the coordinate transformation $s = r \sin (\phi - \theta), t = r \cos (\phi - \theta)$, $\phi = -c^{-1} \ln |u|$
gives $H = u^{-2} W(s^2 + t^2) \exp(-2 c \tan^{-1}(s / t))$,
i.e., the class 7 spacetimes are a special case of the class 3 spacetimes.
The condition (\ref{eq:Hsatisfy}) leads to the following
\be
r W_{,r}  (a'(u) - \Theta) /2 + (a'(u) + \Theta - 2 c \gamma) W = 0 \: ,
\label{eq:Hsatisfycase7}
\ee
and $E'(u) = a'''(u) = c_A(u) = 0$ and the general CKV takes the form (\ref{eq:ckvcase2}). Thus the spacetime will only admit
SCKV. If $W$ is an arbitrary function then it follows immediately from
(\ref{eq:Hsatisfycase7}) that the only CKV which occur are the three KV already listed in \cite{sippel86}, i.e., $k$ and
\be
X_2 = \partial_u \: , \qquad
X_3 = c \left(u\partial_u - v\partial_v \right) - \partial_\theta \: ,
\label{eq:KVsforcase7}
\ee
with Lie brackets
\[
[ k , X_2 ] = 0  \: , \qquad [ k , X_3 ] = -c k \: , \qquad [ X_2, X_3 ] = c X_2 \: .
\]

If $W(r)$ is not arbitrary, the differential equation (\ref{eq:Hsatisfycase7}) gives $W(r)=\delta r^{-\sigma}$, where $\delta$ is a non-zero constant, and if
$\sigma \ne 2$ then $a'(u) = \alpha$ is constant where $\sigma$ is given by $(2 - \sigma) \alpha + ( 2 + \sigma) \Theta = 4 c \gamma$.
This leads to an ${\cal H}_4$ with basis $k$, $X_2$, $X_3$ and the proper HKV
\be
Z = 2 u \partial_u +  r \partial_r
+ ((\sigma - 2)/ 2 c) \partial_\theta \, .
\label{eq:case7HKV}
\ee
Note that $\alpha \ne \Theta$ otherwise $Z$ degenerates into the KV $X_3$.
The Lie brackets are
\[
[ k , Z ] = [ X_3 , Z ] = 0 \; , \qquad [ X_2 , Z ] = 2 \, X_2 \; .
\]
When $\sigma = 2$ we find that $\Theta = c \gamma$ leading to an ${\cal S}_5$ with basis $k$, $X_2$, $X_3$, proper HKV $Z$
and proper SCKV $S$ given by
\ba
Z & = & 2u \partial_u + r \partial_r \; ,
\label{eq:case7HKV2}\\
S & = & u^2 \partial_u + {r^2 \over 2} \partial_v + u r \partial_r \; .
\label{eq:case7SCKV}
\ea
The Lie brackets are
\ba
& & [ k, Z ] = [ k, S ] = [ X_3, Z ] = 0 \; , \qquad [ X_3, S ] = c S \; ,
\nonumber\\
& & { [ Z, S ] = 2 S \; , \qquad [ X_2, Z ] = 2 X_2 \; , \qquad
[ X_2, S ] = Z  \; .}
\ea
Hence the isometry class 7 spacetimes with metric
\be
ds^2 = - 2 du dv - 2 \delta r^{- \sigma} \exp{(2 c \, \theta)} du^2 + dy^2 + dz^2
\label{eq:summarymetric7}
\ee
admits an ${\cal H}_4 \supset {\cal G}_3$ if $\sigma \ne 2$, or an ${\cal S}_5 \supset {\cal H}_4 \supset {\cal G}_3$ if $\sigma = 2$. All other isometry
class 7 spacetimes admit only the ${\cal G}_3$.

\subsection{Isometry class 8}
In this case the function $H$ has the form $H = W(s) \exp (2 \delta w)$,
where $s = \eta z - \sigma y$, $w = (\eta y + \sigma z)$ and $\delta = - \epsilon / (\eta^2 + \sigma^2)$, where $\eta$, $\sigma$ are constants such that
$\eta^2 + \sigma^2 \ne 0$, and $\epsilon$ is a non-zero constant. Sippel and Goenner use the notation $t = \delta w$.
Condition (\ref{eq:Hsatisfy}) leads to $E'(u) = \gamma = \alpha - \Theta = a'''(u) = c''_A(u) = 0$.
If $W$ is an arbitrary function then it follows that the only CKV which occur are the three KV already listed in \cite{sippel86}.
The Lie brackets of the three KVs $k$ and
\be
X_2 = \partial_u \: , \qquad
X_3 = \epsilon \left(u\partial_u - v\partial_v \right)
+ \eta \partial_y + \sigma \partial_z\: ,
\label{eq:KVsforcase8}
\ee
are
\[
[ k , X_2 ] = 0  \: , \qquad [ k , X_3 ] = -\epsilon k \: , \qquad [ X_2, X_3 ] = \epsilon X_2 \: .
\]
Further specialization is possible for $\epsilon \ne 0$ only if $H$ is of the form $W(s) = K \exp( c s)$
for some non-zero constants $K$ and $c$. However, this corresponds to isometry class 9.

\subsection*{Isometry class 8 ($\epsilon = 0$)}

In this case the function $H$ has the form $H = W(s)$. We have $a'''(u) = 0$. If $W(s)$ is arbitrary then $\gamma = c_A''(u) = 0$
and we have a ${\cal G}_4$ consisting of KVs $k$ and
\be
\fl X_2 = \partial_u \: , \qquad
X_3 = \eta \partial_y + \sigma \partial_z \: , \qquad
X_4 = \eta (y \partial_v + u \partial_y) + \sigma (z \partial_v + u \partial_z) \, ,
\label{eq:nov62138}
\ee
with Lie brackets
\ba
{[}k, X_2] = [k, X_4] = [k, X_3] = [X_2, X_3] = 0 \, ,
\nonumber\\
{[}X_2, X_4] = X_3 \, , \qquad [X_3, X_4] =  (\sigma^2 + \eta^2) k \, .
\nonumber
\ea
The  ${\cal G}_3$ subalgebra consisting of $\{ k, X_3, X_4 \}$ is isomorphic to that of isometry class $1i$.
Further specialization is possible if $\gamma = 0$ and
\ba
{[}\psi s + (-\sigma c_2(u) + \eta c_3(u))] W(s)_{,s} + (\sigma^2 + \eta^2)^{-1} (-\sigma c_2(u) + \eta c_3(u))'' s
\nonumber\\
+ E'(u) + (a'(u) + \Theta) W(s) = 0 \, ,
\nonumber
\ea
where $a(u) = \rho u^2 + \alpha u + \beta$ and $\rho$, $\alpha$ and $\beta$ are arbitrary constants.
Let us consider some examples. The restriction $(-\sigma c_2(u) + \eta c_3(u)) = \Theta = 0$ leads to $W(s) = K s^{-2}$ with a
${\cal S}_6 \supset {\cal H}_5 \supset {\cal G}_4$, which is class $Biv$.
The condition $(-\sigma c_2(u) + \eta c_3(u)) = c$, $c$ being a non-zero constant, leads to $W(s) = K \exp(-2 \alpha s / c)$ with a ${\cal G}_5$, which is isometry class 9.
There is an example with a ${\cal H}_5 \supset {\cal G}_4$. The conditions $\psi$ is constant and $(-\sigma c_2(u) + \eta c_3(u)) = 0$ lead to $c_A''(u) = 0$,
$W(s) = K s^c$ and $c = -2 (\alpha + \Theta)/(\alpha - \Theta)$, $\Theta \ne 0$.
Then we have an algebra consisting of the KVs $k$, (\ref{eq:nov62138}) and the HKV
\[
Z = 2 v \partial_v + y \partial_y + z \partial_z \: ,
\]
with Lie brackets
\[
[k, Z] = 2 k \, , \qquad [X_2, Z] = 0 \, , \qquad [X_3, Z] = X_3 \, , \qquad [X_4, Z] = X_4 \, .
\]

\subsection{Isometry class 9}
This is a special case of isometry class 8, the function $H$ having the form $H = K \exp{2[\eta y - \sigma z]}$,
where $K$, $\eta$ and $\sigma$ are constants, $\eta^2 + \sigma^2 \ne 0$.
Condition (\ref{eq:Hsatisfy}) leads to the following equations, by equating coefficients of $y^2$, $z^2$, $y$, $z$, etc,
\ba
& & a'''(u) = 0 \: , \qquad c''_A(u) = 0 \: , \qquad E'(u) = 0 \: ,
\nonumber\\
& & 2 \eta c_y(u) - 2 \sigma c_z(u) = -(a'(u) + \Theta) \: , \nonumber\\
& & \eta (a'(u) - \Theta) + 2 \gamma \sigma  = 0 \: , \qquad \sigma (a'(u) - \Theta)  - 2 \gamma \eta  = 0 \: . \nonumber
\ea
We find that the last two conditions demand that $a'(u) = \Theta$ and so $\psi = \gamma = 0$.
Thus only KV's are possible in this spacetime. The ${\cal G}_5$ consists of $k$ and
\ba
X_2 & = & \partial_u  \: ,
\qquad
X_3 = \partial_z + \sigma \left( u \partial_u - v \partial_v \right) \: , \nonumber\\
X_4 & = & \partial_y - \eta \left( u \partial_u - v \partial_v \right)
\: , \qquad
X_5 = u \left( \eta \partial_z + \sigma\partial_y \right)
+ (\sigma y + \eta z ) \partial_v \: ,
\label{eq:KVcase9}
\ea
and the Lie brackets are
\ba
& & [ k, X_4 ] = \eta k \: , \qquad [ k, X_3 ] = -\sigma k \: , \qquad [k, X_2 ] = 0 \: , \qquad [k, X_5 ] = 0 \: ,
\nonumber\\
& & [X_4 , X_2 ] =  \eta X_2 \: ,  \qquad [X_3 , X_2 ] =  -\sigma X_2 \: , \qquad [X_2 , X_5 ] = \eta X_3 + \sigma X_4 \: . \nonumber\\
& & [X_4 , X_3 ] = 0 \: , \qquad
[X_3, X_5] = \eta k + \sigma X_5 \: , \qquad [X_4, X_5] = \sigma k - \eta X_5 \: . \nonumber
\ea

\begin{table}
\caption{The conformal symmetry classes for the plane wave spacetimes. The quantities $a,b,c,l, \beta$, $\gamma$ and $\epsilon$ are constants. Sippel and Goenner state that classes 11 and 13 cannot be vacuum, yet we obtain vacuum solutions by putting $a+c=0$ in both cases. However, the latter are special cases of cases 12 and 14 respectively.}
\footnotesize\rm
\begin{tabular*}{\textwidth}{@{}l*{15}{@{\extracolsep{0pt plus12pt}}l}}
\br

Class & Algebra & Orbits & Metric function $H$ & $F$\\ \hline

$10$  & ${\cal H}_6 \supset {\cal G}_5$ & $3 n$ & $(A(u) y^2 + C(u) z^2)/2 + B(u) yz$ & $A(u) + C(u)$ \\

$11$  & ${\cal H}_7 \supset {\cal G}_6$ & 4 & $A(u) = au^{-2}$, $B(u) = bu^{-2}$, & $(a+c) u^{-2}$\\
 & & &  $C(u) = cu^{-2}$  & \\

$12$  & ${\cal H}_7 \supset {\cal G}_6$ & 4 & $A(u) = cu^{-2}(\sin \phi - l)$, $B(u) = c u^{-2} \cos \phi$, & $-2 \, cl u^{-2}$ \\
 & & & $C(u) = -cu^{-2}(\sin \phi + l)$, $\phi = 2\epsilon \ln |u|$ & \\

$13$  & ${\cal H}_7 \supset {\cal G}_6$ & 4 & $A(u)=a$, $B(u)=b$, $C(u)=c$, & $a + c$ \\

$14$  & ${\cal H}_7 \supset {\cal G}_6$ & 4 & $A(u) = c\sin \phi + l$, $B(u) = c \cos \phi$, & $2 \, l$ \\
 & & & $C(u) =  -c\sin \phi + l$, $\phi = 2\epsilon u$ & \\

$10i$ & ${\cal S}_7 \supset {\cal H}_6 \supset {\cal G}_5$ & 4 & $A(u)=  c (u^2 + \beta)^{-2} (\sin \phi + l)$ & $2 \, c \, l (u^2 + \beta)^{-2}$\\
 & & & $B(u) = c (u^2 + \beta)^{-2} \cos \phi$  & \\
 & & & $C(u) = c (u^2 + \beta)^{-2} (-\sin \phi + l)$,  & \\
 & & & $\phi = 2 \gamma \int (u^2 + \beta)^{-1} du$  & \\

$10ii$ & ${\cal S}_7 \supset {\cal H}_6 \supset {\cal G}_5$ & 4 & $A(u) = -a (u^2 + \beta)^{-2}$, $B(u) = - b(u^2 + \beta)^{-2}$ & $-(a+c) (u^2 + \beta)^{-2}$\\
 & & & $C(u) = - c (u^2 + \beta)^{-2}$  & \\

$10iii$ & ${\cal C}_7 \supset {\cal H}_6 \supset {\cal G}_5$ & 4 & $a'''(u) \ne 0$ & $A(u) + C(u)$ \\

\br
\end{tabular*}
\label{tab:planewaves}
\end{table}

\subsection{Isometry class 10: plane wave spacetimes}
\label{sec:genplane}
We note that when the function $H$ has the form
\be
2H = A(u) y^2 + 2B(u) yz + C(u) z^2 \: ,
%%%\label{eq:H}
\ee
then (\ref{eqn:ppw}) represents an Einstein-Maxwell spacetime and admits at least an ${\cal H}_6$. Such a spacetime is
referred to as a {\it plane wave spacetime}.
When $A(u) = - C(u)$ the spacetime is vacuum and when $A(u) = C(u)$ and $B(u)=0$ the spacetime is conformally flat.
Since the plane wave admits an ${\cal H}_6$ we can put $\Theta=0$ since this parameter
just corresponds to the addition of an HKV given by $Z$ of equation (\ref{eq:vectorZ}).
Thus, we need only consider the proper CKV given by
\ba
Y^u & = & a(u) \: , \nonumber\\
Y^v & = & a''(u) \, x_A x^A/4 + c'_B(u) \, x^B + E(u)  \: , \nonumber\\
Y^A & = & a'(u) \, x^A / 2 + \gamma \, \epsilon_{AB} x^B + c_A(u) \: ,
\ea
with conformal scalar $\psi = a'(u) /2$. The condition (\ref{eq:Hsatisfy}) allows us to arrive at the
following conclusions.
Condition (\ref{eq:Hsatisfy}) leads to the following equations, by equating coefficients of $y^2$, $z^2$, $yz$, $y$, $z$ respectively and
functions of $u$ only
\ba
a'''(u)/4 + A'(u) a(u)/2 + A(u) a'(u) - \gamma B(u) = 0 \label{eq:yy} \: , \\
a'''(u)/4 + C'(u) a(u)/2 + C(u) a'(u) + \gamma B(u) = 0 \label{eq:zz} \: , \\
B'(u) a(u) + 2 a'(u) B(u) + \gamma (A(u) - C(u)) = 0 \label{eq:yz} \: , \\
c''_y(u) + A(u) c_y(u) + B(u) c_z(u) = 0 \label{eq:y} \: , \\
c''_z(u) + B(u) c_y(u) + C(u) c_z(u) = 0 \label{eq:z} \: , \\
E'(u) = 0 \: . \label{eq:u}
\ea
Equations (\ref{eq:y}) and (\ref{eq:z}) correspond to the condition on $H$ to admit the four KV listed in
case 10 of \cite{sippel86}
\be
X_i = c_A(u)_i \, \partial_A + c'_A(u)_i \, x^A \, \partial_v \; ,
\label{eq:KVcase10}
\ee
$i= 2, \dots , 5$, and so we can put $c_A(u) = 0$ for the remaining seventh CKV. Also, $E=constant$ just corresponds to the addition of the KV $k$ and so we shall put $E=0$. Thus the most general CKV is given by
\be
Y = a(u) \partial_u + \case{1}{4} a''(u) \, x^A x_A \partial_v
+ ( \case{1}{2} a'(u) \, x^A  + \gamma \, \epsilon_{AB} x^B ) \partial_A \: .
\label{eq:CKVcase10}
\ee
The Lie brackets of the seven CKV are
\ba
\fl [ k, Z ] = 2k  \: , \qquad [ k, X_i ] = 0  \: , \qquad [ Z , X_i ] = - X_i  \: ,
\qquad {[ X_i, X_j ]} = 2 \, Q_{[ij]} \, k   \: ,
\nonumber\\
\fl [ k, Y ]  = 0  \: , \qquad  [ Z, Y ]  = 0   \: , \qquad  [ Y, X_i ]  = m \, X_i  \: , \ea
where $Q_{ij} =  \delta^{AB} \, c_A(u)_i \, c'_B(u)_j$, and so $Q_{[ij]}$ is a constant, and $m$ is a constant.
Equation (\ref{eq:yz}) gives
\be
\gamma = (a(u)B'(u) + 2 a'(u) B(u) / (C(u) - A(u))
\label{eq:braineqn0}
\ee
which implies that the rhs of this equation must be a constant, and exclusion of type $O$ gives $A(u)-C(u) \ne 0$. Substituting this equation into (\ref{eq:yy}) and (\ref{eq:zz}) we obtain
\ba
& & (A(u) - C(u)) a'''(u)/4 + [A(u) (A(u) - C(u)) + 2B^2(u)]a'(u) \nonumber\\
& & +[A'(u) (A(u) - C(u)) /2 + B(u) B'(u)] a(u) = 0 \: , \label{eq:braineqn1} \\
& & (A(u) - C(u)) a'''(u)/4 + [C(u) (A(u) - C(u)) - 2B^2(u)]a'(u) \nonumber\\
& & +[C'(u) (A(u) - C(u)) /2 - B(u) B'(u)] a(u) = 0 \: . \label{eq:braineqn2}
\ea
Subtracting (\ref{eq:braineqn1}) and (\ref{eq:braineqn2}) gives
\ba
& & 2 [(A(u) - C(u))^2 + 4B^2(u)] a'(u) \nonumber\\
& & + [(A(u) - C(u)) (A'(u) - C'(u)) + 4B(u) B'(u)] a(u) = 0
\nonumber
\ea
which integrates to give
\be
a(u) = L | (A(u) - C(u))^2 + 4B^2(u) | ^{-{1 \over 4}}
\label{eq:braineqn4}
\ee
where $L$ is an arbitrary constant. Adding (\ref{eq:braineqn1}) and (\ref{eq:braineqn2}) gives
\be
a'''(u) + 2 (A(u) + C(u)) a'(u) + (A'(u) + C'(u)) a(u) = 0 \: .
\label{eq:braineqn5}
\ee
Substituting (\ref{eq:braineqn4}) into (\ref{eq:braineqn5}) gives two conditions to be satisfied by the functions $A(u)$, $B(u)$, $C(u)$ and their derivatives in order for a CKV to exist. We note that Classes 11-14 in table \ref{tab:planewaves} admit an ${\cal H}_7$ and
so can admit no further symmetries. We present other examples of plane wave spacetimes admitting proper SCKV or non-special CKV.

\noindent {\bf Example (i)} The function $H$ has the form
\ba
A(u) & = & c (u^2 + \beta)^{-2} (\sin \phi + l) \: \: , \qquad B(u)=  c (u^2 + \beta)^{-2} \cos \phi \: , \nonumber\\
C(u) & = & c (u^2 + \beta)^{-2} (-\sin \phi + l) \: ,
\label{eq:pp+SCKV}
\ea
where
\be
\phi = 2 \gamma \int (u^2 + \beta)^{-1} du
\label{eq:defnphi}
\ee
and $c$, $\beta$ and $\gamma$ are constants. The proper SCKV is
\be
S = (u^2 + \beta) \partial_u + \case{1}{2} (y^2 + z^2)\partial_v
+ u \left( y \partial_y + z \partial_z \right)
+ \gamma \left( z \partial_y - y \partial_z \right)
\label{eq:sckv}
\ee
with conformal scalar $\psi=u$. This spacetime is vacuum iff $l=0$.

\noindent {\bf Example (ii)} The function $H$ has the form
\ba
A(u)= m (u^2 + \beta)^{-2} \: , \qquad B(u)=  n (u^2 + \beta)^{-2} \: ,
\nonumber\\
C(u)= p (u^2 + \beta)^{-2} \: ,
\label{eq:stgammazero}
\ea
where $m$, $n$ and $p$ are constants. The corresponding proper SCKV is
\be
S = (u^2 + \beta) \partial_u + \case{1}{2} (y^2 + z^2)\partial_v
+ u \left( y \partial_y + z \partial_z \right) \: .
\label{eq:sckvgammazero}
\ee
Setting $\gamma=0$ in (i) leads to classes of functions of $H$ more restricted than permitted in (ii).
This spacetime is vacuum iff $m=-p$.

The following are examples with non-special CKV. We substitute (\ref{eq:braineqn4}) into (\ref{eq:braineqn5}) and demand
$a'''(u) \ne 0$.

\noindent {\bf Example (iii)}
\ba
A(u) = (\alpha + 1) u^2 / 2 - 3u^{-2} / 4 \: , \qquad B(u) = 0 \: ,
\nonumber\\
C(u) = (\alpha - 1) u^2 / 2 - 3u^{-2} / 4 \: , \qquad
a(u) = u^{-1} \: ,
\nonumber
\ea
and the CKV and conformal scalar are
\[
Y = u^{-2}
( u \partial_u + \case{1}{2} x^A x_A  \partial_v
- \case{1}{2} x^A \partial_A ) \: , \qquad
\psi = -u^{-2} / 2 \: ,
\]
where $\alpha$ is a constant.

\noindent {\bf Example (iv)}

\ba
A(u) & = & (\alpha + 1) e^{-2u} - \case{1}{4}  \: , \qquad B(u) = \beta e^{-2u}  \: , \nonumber\\
C(u) & = & (\alpha - 1) e^{-2u} - \case{1}{4} \: , \qquad a(u) = e^u \: ,
\ea
where $\alpha$ and $\beta$ are constants. The CKV and conformal scalar are
\[
Y = e^u ( \partial_u + (r^2 / 4) \partial_v + (r/2) \partial_r) \: , \qquad \psi = e^u / 2 \: .
\]

Of course, every plane wave spacetime is conformally related to a vacuum plane wave spacetime \cite{hall90}.

\section{Conformally related pp-wave spacetimes}

\label{sec:typeNdual}
The transformation (44) in \cite{tupper03} also applies to general pp-wave spacetimes.

\begin{thm} \label{propconf}
Given an arbitrary pp-wave spacetime (\ref{eqn:ppw}), then the conformally related spacetime
\be
d{\bar s}^2 = \Omega^2(u) ds^2
\label{eq:confRelatedePpwave}
\ee
is also a pp-wave spacetime. The coordinate transformation
\ba
u & = & \int (\Omega({\bar u}))^{-2} d{\bar u} \: , \qquad v = {\bar v} - (\ln \Omega({\bar u}))_{,{\bar u}}
({\bar y}^2 + {\bar z}^2) / 2 \: , \nonumber\\
y & = & {\bar y} \, / \, \Omega({\bar u}) \: ,  \qquad z = {\bar z} \, / \, \Omega({\bar u}) \: ,
\label{eq:conftransf}
\ea
(where, with an abuse of notation, $\Omega({\bar u})$ is the function $\Omega(u)$ written in terms of the new coordinate ${\bar u}$)
allows us to put the metric (\ref{eq:confRelatedePpwave}) into the form of a pp-wave spacetime
\be
d{\bar s}^2 = -2 {\bar H} ({\bar u}, {\bar y}, {\bar z}) d{\bar u}^2 - 2 d{\bar u} d{\bar v} + d{\bar y}^2 + d{\bar z}^2 \: .
\ee
The new metric function ${\bar H}$ is given by
\ba
2 {\bar H} ({\bar u}, {\bar y}, {\bar z}) & = &
2H ({\bar u}, {\bar y}, {\bar z}) \Omega^{-2}({\bar u})
 - \Omega^{-1}({\bar u}) \Omega_{,{\bar u}{\bar u}}({\bar u}) ({\bar y}^2 + {\bar z}^2) \: ,
\label{eq:newH}
\ea
where $H ({\bar u}, {\bar y}, {\bar z})$ is the function $H(u,y,z)$ written in terms of the new coordinates
${\bar u}, {\bar y}, {\bar z}$.
\end{thm}
\begin{table}
\caption{Conformally related spacetimes. The possible Lie algebra types are listed.}
\footnotesize\rm
\begin{tabular*}{\textwidth}{@{}l*{15}{@{\extracolsep{0pt plus12pt}}l}}
\br
Type & Algebras \\ \hline
$C3$ & ${\cal C}_2 \supset {\cal G}_1$, ${\cal S}_2 \supset {\cal G}_1$, ${\cal H}_2$  \\

$C4$ & ${\cal C}_2 \supset {\cal G}_1$, ${\cal S}_2 \supset {\cal G}_1$, ${\cal H}_2$   \\

$C5$ & ${\cal C}_3 \supset {\cal G}_2$, ${\cal S}_3 \supset {\cal G}_2$, ${\cal H}_3$   \\

$C5i$ & ${\cal C}_4 \supset {\cal G}_2$, ${\cal C}_4 \supset {\cal S}_3 \supset {\cal G}_2$,
${\cal C}_4 \supset {\cal G}_3$, ${\cal C}_4 \supset {\cal H}_3$  \\

$C5ii$ & ${\cal C}_4 \supset {\cal G}_2$, ${\cal C}_4 \supset {\cal G}_3$,
${\cal C}_4 \supset {\cal H}_3$, ${\cal H}_4$, ${\cal C}_4
\supset {\cal S}_3 \supset {\cal G}_2$  \\

$C6$ & ${\cal C}_3 \supset {\cal G}_2$, ${\cal S}_3 \supset {\cal G}_2$, ${\cal H}_3$   \\

$C6i$ & ${\cal C}_5 \supset {\cal G}_2$, ${\cal C}_5 \supset {\cal H}_4$,
${\cal C}_5 \supset {\cal S}_4 \supset {\cal G}_3$, ${\cal C}_5 \supset {\cal S}_3 \supset {\cal G}_2$   \\

$C6ii$ & ${\cal C}_5 \supset {\cal G}_2$, ${\cal S}_5 \supset {\cal G}_3$   \\

$C6iii$ & ${\cal C}_4 \supset {\cal G}_2$, ${\cal C}_4 \supset {\cal S}_3 \supset {\cal G}_2$,
${\cal C}_4 \supset {\cal G}_3$, ${\cal C}_4 \supset {\cal H}_3$, ${\cal S}_4 \supset {\cal G}_3$,
${\cal S}_4 \supset {\cal H}_3$ \\

$C6iv$ & ${\cal C}_4 \supset {\cal G}_2$, ${\cal C}_4 \supset {\cal S}_3 \supset {\cal G}_2$,
${\cal C}_4 \supset {\cal G}_3$, ${\cal C}_4 \supset {\cal H}_3$, ${\cal S}_4 \supset {\cal G}_3$,
${\cal S}_4 \supset {\cal H}_3$ \\

$C7$ & ${\cal C}_3 \supset {\cal G}_2$, ${\cal S}_3 \supset {\cal G}_2$, ${\cal H}_3$ \\
$C7i$ & ${\cal C}_4 \supset {\cal G}_2$, ${\cal S}_4 \supset {\cal G}_3$, ${\cal S}_4 \supset {\cal H}_3$ \\
$C7ii$ & ${\cal C}_5 \supset {\cal G}_2$, ${\cal S}_5 \supset {\cal G}_3$ \\
C8 & ${\cal C}_3 \supset {\cal G}_1$, ${\cal C}_3 \supset {\cal H}_2$, ${\cal S}_3 \supset {\cal H}_2$,
${\cal C}_3 \supset {\cal S}_2 \supset {\cal G}_1$ \\
$C8$($\epsilon = 0$) &  ${\cal C}_4 \supset {\cal G}_3$, ${\cal S}_4 \supset {\cal G}_3$, ${\cal H}_4$ \\
$C9$ & ${\cal C}_5 \supset {\cal G}_3$, ${\cal S}_5 \supset {\cal H}_4$ \\

\br
\end{tabular*}
\label{tab:generalConfTransf}
\end{table}

The CKV of $ds^2$ become CKV of $d{\bar s}^2$ with the same conformal Lie algebra,
the conformal scalars being related by ${\bar \psi} = (\ln \Omega)_{,u} Y^u + \psi$. Those CKV of $ds^2$ with $Y^u = 0$ will have their
conformal scalars unchanged, in particular the KV $k$ will remain a KV for the new spacetime.
Table \ref{tab:generalConfTransf} lists the possible conformal symmetry Lie algebras arising from this conformal process.
The general equation for the new conformal scalar is
\be
{\bar \psi} = ({\bar a}({\bar u})_{,{\bar u}} - {\bar \Theta})/2 = a(u) (\ln \Omega(u))_{,u} +
(a'(u) - \Theta)/2,
\label{eq:general}
\ee
where $\Omega^2(u) du = d {\bar u}$. We wish to solve for the functional form $\Omega(u)$ which will give particular conformal
symmetries. To obtain an SCKV we require
${\bar a}({\bar u}) = {\bar \rho} {\bar u}^2 + {\bar \alpha} {\bar u} + {\bar \beta}$, with
${\bar \rho}, {\bar \alpha}, {\bar \beta}$ constants. This gives
$(2{\bar \rho} {\bar u} + {\bar \alpha} - {\bar \Theta})/2 = a(u) (\ln \Omega(u))_{,u} +
(a'(u) - \Theta)/2$. In order to solve for $\Omega(u)$ we use the substitution
$\Omega^2(u) du = d {\bar u}$ in the lhs. For the case of a proper SCKV we can differentiate with respect to $u$ to get
\be
{\bar \rho} \Omega^2(u) = a'(u) (\ln \Omega(u))_{,u} + a(u) (\ln \Omega(u))_{,uu} + a''(u)/2.
\label{eq:SCKV}
\ee
Equation (\ref{eq:SCKV}) is nonlinear in $\Omega(u)$. It is difficult to say anything in general about
the solutions to this, even the number of independent solutions. However, it can be written in the form
\[
\int {d \omega \over ({\bar \rho} \omega^2 + \mu \omega + \nu)} = \int a^{-1}(u) du \, , \qquad \omega_{,u} = \Omega^2(u) \, ,
\]
where $\mu, \nu$ are constants.
For an HKV (not necessarily proper HKV) we do not need to do the differentiation, we just take (\ref{eq:general})
as it is to get
\be
({\bar \alpha} - {\bar \Theta})/2 = a(u) (\ln \Omega(u))_{,u} + (a'(u) - \Theta)/2 \, .
\ee
This last equation gives, for $a(u) \ne 0$
\be
(\ln \Omega(u))_{,u} = [({\bar \alpha} - {\bar \Theta}) - (a'(u) - \Theta) ] / 2 \, a(u)
\ee
which gives, for $a'(u) \ne 0$
\be
\Omega(u) = a^{- 1 / 2}(u) \exp \left[{1 \over 2} ({\bar \alpha} - {\bar \Theta} + \Theta) \int a^{-1}(u) du \right] \, .
\ee

If a further spacetime is formed by applying this conformal technique to $d{\bar s}^2$ then in general this spacetime will not be the original spacetime $ds^2$. However, if we choose $\Omega(u) = u^{-1}$ then any two conformally related spacetimes will be `dual' to each other, i.e., repeated application of this conformal factor will just transform one spacetime into the other.
In this case the coordinate transformation (\ref{eq:conftransf}) reduces to
\[
{\bar u} = -u^{-1} \; , \qquad  {\bar v} = v - u^{-1} (y^2 + z^2)/2 \; , \qquad {\bar y} = u^{-1} y \; , \qquad {\bar z} = u^{-1} z \; .
\]
The individual conformal scalars are related by ${\bar \psi} = -u^{-1} Y^u + \psi$. We now consider these `dual' spacetimes in detail. The conformal scalar corresponding to the general SCKV
(\ref{eq:generalSCKV}) is ${\bar \psi} = \beta {\bar u} - (\alpha + \Theta)/2$ leading to an SCKV.
It is straightforward to verify that each of the classes $A$, $B$, $C$ and $D$ in section \ref{sec:typeN}, the new spacetimes
are in the same class. We determine the new spacetimes and conformal symmetries corresponding to the isometry classes
in section \ref{sec:SandG}, and their specializations. We drop the bars on the coordinates. The dual spacetimes are
listed in table \ref{tab:dual}.
\begin{table}
\caption{The conformal symmetry classes of the dual pp-wave spacetimes. Details are given only for spacetimes which have not been listed previously, i.e., for those which are new or
{\it specializations} of conformal symmetry classes in tables 1 - 3.}
\footnotesize\rm
\begin{tabular*}{\textwidth}{@{}l*{15}{@{\extracolsep{0pt plus12pt}}l}}
\br

Dual & Class & Algebra &  Metric function $H$ & $F$\\ \hline

$D3$ & $1$  & ${\cal H}_2$ & $W(u^{-1} s, u^{-1} t)$  & $W_{,ss} + W_{,tt}$ \\

$D4$ & $1$ & ${\cal S}_2$ & $u^{-2}W(u^{-1} s, u^{-1} t)$  & $u^{-2}(W_{,ss} + W_{,tt})$\\

$D5$ & $2$ & ${\cal H}_3$  & $W(u^{-1} r)$  & $W_{,rr} + r^{-1}W_{,r}$\\

$D5i$ & $6iii$  & & & \\

$D5ii$ & $2$ & ${\cal C}_4 \supset {\cal H}_3 $  & $ \delta (-u)^\sigma r^{-\sigma} - \sigma (2 - \sigma)^2 u^{-2} r^2 $
& $\delta (-u)^\sigma \sigma^2 r^{-(\sigma + 2)} - 4 \sigma (2 - \sigma)^2 u^{-2}$ \\

$D6$ & $2$ & ${\cal S}_3 \supset {\cal G}_2$ & $u^{-2} W(u^{-1} r)$  & $u^{-2}(W_{,rr} + r^{-1}W_{,r})$\\

$D6i$ & $Ci$ & ${\cal C}_5 \supset {\cal S}_3 \supset {\cal G}_2$ & $\delta r^{-2} + \nu u^{-4} r^2 / 4$  &
$4 \delta r^{-4} + \nu u^{-4}$\\

$D6ii$ & $6ii$  & & & \\

$D6iii$ & $5i$ & & & \\

$D6iv$ & $2$ & ${\cal S}_4 \supset {\cal H}_3$  & $\delta (-u)^{\sigma - 2} r^{-\sigma}$  &
$\delta (-u)^{\sigma - 2} \sigma^2 r^{-(\sigma + 2)}$ \\

$D7$ & $1$ & ${\cal S}_3 \supset {\cal H}_2$ & $u^{-2} W(u^{-1} r) \, e^{2 c \theta}$  &
$u^{-2} e^{2 c \theta} [W_{,rr} + r^{-1} W_{,r} + 4 c^2 W]$ \\

$D7i$ & $3$ & ${\cal S}_4 \supset {\cal H}_3$ & $\delta (-u)^{\sigma - 2} r^{-\sigma} e^{2 c \theta}$  &
$(\sigma^2 + 4 c^2) r^{-2} H$ \\

$D7ii$ & $7ii$ & & & \\

$D8$ & $1$ & ${\cal S}_3 \supset {\cal H}_2$ & $u^{-2} W(u^{-1} s) \, e^{-2 t / u}$  &
$u^{-2} e^{-2t/u} [W_{,ss} + 4 u^{-2} W]$ \\

$D8(\epsilon = 0)$ & $1i$ & ${\cal S}_4 \supset {\cal G}_3$ & $u^{-2} W(u^{-1} s) $  & $(\eta^2 + \sigma^2) u^{-2} W_{,ss}$ \\

$D8(\epsilon = 0)i$ & $1i$ & ${\cal S}_5 \supset {\cal H}_4$ & $u^{-2} K (-u^{-1} s)^c  $  & $(\eta^2 + \sigma^2) c(c-1) K u^{-2} (-u)^{-c} s^{c-2}$ \\

$D9$ & New & ${\cal S}_5 \supset {\cal H}_4$ & $u^{-2} K \exp [2 u^{-1} (\sigma z - \eta y)]$  &
$4 u^{-2} (\eta^2 + \sigma^2) H$\\

\br
\end{tabular*}
\label{tab:dual}
\end{table}
\subsection{Dual of isometry class 1$i$}
The new dual spacetime also belongs to the same class, i.e., ${\cal G}_3 \mapsto {\cal G}_3$. The vacuum condition implies type $O$ and hence flatness.

\subsection{Dual of isometry class 2}
The new dual spacetime also belongs to the isometry class 2 and so the KV $X_2$ (\ref{eq:KVsforcase2}) remains a KV for the new spacetime, i.e., ${\cal G}_2 \mapsto {\cal G}_2$.
The duals of the vacuum spacetimes are also vacuum.

\subsection{Dual of isometry class 3}
The new function is $H=W(u^{-1}s, u^{-1}t)$, $s=y \sin \phi - z \cos \phi$, $t=y \cos \phi + z \sin \phi$,
$\phi=- \epsilon \ln |-u|$.
The KV $X_2$ (\ref{eq:KVsforcase3}) becomes an HKV given by
\be
X_2 = \epsilon \left(z \partial_y - y \partial_z \right)
- \left( y \partial_y + z \partial_z + u \partial_u + v \partial_v  \right) \, ,
\label{eq:dualclass3HKV}
\ee
in the new spacetime and so we have ${\cal G}_2 \mapsto {\cal H}_2$, i.e., the dual
spacetime is of isometry class 1 admitting a proper HKV.

\subsection{Dual of isometry class 4}

The new function is $H= u^{-2} W(u^{-1}s, u^{-1}t)$, $s=y \sin \phi - z \cos \phi$, $t=y \cos \phi + z \sin \phi$,
$\phi=- \epsilon u^{-1}$.
The KV $X_2$ (\ref{eq:KVsforcase4}) becomes an SCKV given by
\be
X_2 = \epsilon \left( z \partial_y - y \partial_z \right)
+ u \left( y \partial_y + z \partial_z \right) + u^2 \partial_u + \case{1}{2} (y^2 + z^2) \partial_v  \, ,
\label{eq:dualclass4SCKV}
\ee
in the new spacetime and so we have ${\cal G}_2 \mapsto {\cal S}_2$, i.e., the dual
spacetime is of isometry class 1 admitting a proper SCKV.

\subsection{Dual of isometry class 5}

The new function is $H= W(u^{-1}r)$.
The KV $X_3$ in (\ref{eq:KVsforcase5}) is unchanged and the KV $X_2$ in (\ref{eq:KVsforcase5}) becomes an HKV given by
\be
X_2 = - \left( x^A \partial_A + u \partial_u + v \partial_v  \right) \, ,
\label{eq:dualclass5HKV}
\ee
in the new spacetime. Thus we have ${\cal G}_3 \mapsto {\cal H}_3$ and the dual spacetime is
of isometry class 2 admitting a proper HKV.

For the ${\cal S}_4$ specialization given by equation (\ref{eq:summarymetric5i}) the new metric function is given by
$H = \zeta \ln |- u^{-1} r|$ and the coordinate transformation ${\bar v} = v - \zeta u ( \ln |-u| - 1)$ transforms $H$ into
$H = \zeta \ln |r|$, which is the class 6 metric specialization (\ref{eq:summarymetric6iii}).
The SCKV $S$ (\ref{eq:case5SCKV1}) becomes the KV
\be
S = \partial_u\, ,
\label{eq:dualclass5KV}
\ee
while $X_2$ becomes the HKV $Z$ in (\ref{eq:case6HKV}). Thus we have ${\cal S}_4 \mapsto {\cal H}_4$.

For the ${\cal C}_4$ specialization given by (\ref{eq:summarymetric5ii}) the new function is given by
\be
H = \delta (-u)^\sigma r^{-\sigma} - \sigma (2 - \sigma)^2 u^{-2} r^2 \, ,
\label{eq:dualclass5Hiii}
\ee
and the non-special proper CKV becomes the non-special proper CKV
\be
C_1 = (-u)^q \,
[ u^2 \partial_u - (\sigma / (2 - \sigma)) u r \partial_r
+( \sigma (\sigma + 2)(\sigma - 2)^{-2} / 2 ) r^2 \partial_v ],
\label{eq:dualclass5CKV}
\ee
and we have ${\cal C}_4 \mapsto {\cal C}_4 \supset {\cal H}_3$. The dual spacetime is an isometry class 2 specialization.

\subsection{Dual of isometry class 6}

The new function is $H= u^{-2} W(u^{-1}r)$ and in general is an isometry class 2 spacetime.
The KV $X_3$ in (\ref{eq:KVsforcase6}) is unchanged, and the KV $X_2$ in (\ref{eq:KVsforcase6}) becomes an SCKV
\be
X_2 = u r \partial_r + u^2 \partial_u + \case{1}{2} r^2 \partial_v  \, ,
\label{eq:dualclass6SCKV}
\ee
in the new spacetime. Thus we have ${\cal G}_3 \mapsto {\cal S}_3$, where ${\cal S}_3 \supset {\cal G}_2$.

For the ${\cal C}_5 \supset {\cal G}_3$ specialization given by (\ref{eq:summarymetric6i}), the new metric function is
\[
H = \case{1}{4} \nu u^{-4} r^2 + \delta r^{-2}
\]
which is a special case (not listed) of the case $Ci$ of section \ref{sec:typeN} with ${\cal C}_5 \supset {\cal S}_3 \supset {\cal G}_2$.
The CKV (\ref{eq:case6CKV1}) and (\ref{eq:case6CKV2}) become
the non-special proper CKV
\ba
C_1 & = & u^2 \cos (-l u^{-1}) {\partial_u}
+ ( u \cos (-l u^{-1}) - (l / 2) \sin (-l u^{-1}) ) r {\partial_r} \nonumber\\
& & + [ \case{1}{4} (2  - l^2 u^{-2})  \cos (-l u^{-1})  - (l / 2) u^{-1}  \sin (-l u^{-1}) ] r^2 {\partial_v} \, , \\
\label{eq:dual6specCKV1}
C_2 & = & u^2 \sin (-l u^{-1}) {\partial_u}
+ ( u \sin (-l u^{-1}) + (l / 2) \cos (-l u^{-1}) ) r {\partial_r} \nonumber\\
& & + [ \case{1}{4} (2  - l^2 u^{-2})  \sin (-l u^{-1})  + (l / 2) u^{-1}  \cos (-l u^{-1}) ] r^2 {\partial_v} \, .
\label{eq:dual6specCKV2}
\ea
The CKV (\ref{eq:case6CKV3}) and (\ref{eq:case6CKV4}) become the non-special proper CKV
\ba
C_1 & = & u^2 \cosh (-l u^{-1}) {\partial_u}
+ ( u \cosh (-l u^{-1}) + (l / 2) \sinh (-l u^{-1}) ) r {\partial_r} \nonumber\\
& & + [ \case{1}{4} (2  + l^2 u^{-2})  \cosh (-l u^{-1})  + (l / 2) u^{-1}  \sinh (-l u^{-1}) ] r^2 {\partial_v} \, , \\
\label{eq:dual6specCKV3}
C_2 & = & u^2 \sinh (-l u^{-1}) {\partial_u}
+ ( u \sinh (-l u^{-1}) + (l / 2) \cosh (-l u^{-1}) ) r {\partial_r} \nonumber\\
& & + [ \case{1}{4} (2  + l^2 u^{-2})  \sinh (-l u^{-1})  + (l / 2) u^{-1}  \cosh (-l u^{-1}) ] r^2 {\partial_v} \, .
\label{eq:dual6specCKV4}
\ea
For the ${\cal S}_5$ specialization given by (\ref{eq:summarymetric6ii}) the new function is identical to the original and the SCKV $S$ and the KV $X_2$
transform into each other.
For the ${\cal H}_4$ specialization given by (\ref{eq:summarymetric6iii}) the new function is $H = \zeta u^{-2} \ln |-u^{-1} r|$
and the coordinate transformation ${\bar v} = v + \zeta u^{-2} \ln|-u|$ transforms the dual metric into the isometry class 5 specialization
(\ref{eq:summarymetric5i}). Thus the spacetimes with metrics (\ref{eq:summarymetric5i}) and (\ref{eq:summarymetric6iii}) are duals of each other.

For the ${\cal H}_4$ specialization given by (\ref{eq:summarymetric6iv}) the new function is $H = \delta (-u)^{\sigma - 2} r^{-\sigma}$. The HKV $Z$
given by (\ref{eq:case6HKVlast}) gives rise to the HKV
\be
-2 Z = (\sigma + 2) u \partial_u + (\sigma - 2) v \partial_v + \sigma r \partial_r \, .
\ee
The dual spacetime is an isometry class 2 specialization with ${\cal S}_4 \supset {\cal H}_3$.

\subsection{Dual of isometry class 7}
The new function is $H= u^{-2} W(u^{-1}r) e^{2 c \theta}$.
The KV $X_2$ in (\ref{eq:KVsforcase7}) becomes an SCKV given by
\be
X_2 = u r \partial_r + u^2 \partial_u + \case{1}{2} r^2 \partial_v  \, ,
\label{eq:dualclass7SCKV}
\ee
and KV $X_3$ in (\ref{eq:KVsforcase7}) becomes an HKV given by
\be
X_3 = - c \left( r \partial_r + u \partial_u + v \partial_v  \right)
- \partial_\theta \, .
\label{eq:dualclass7HKV}
\ee
Thus we have ${\cal G}_3 \mapsto {\cal S}_3 \supset {\cal H}_2$, i.e., in general the dual spacetime is an isometry class 1 specialization.

For the ${\cal H}_4$ specialization the new function is given by
\be
H =  \delta (-u)^{\sigma - 2} r^{ - \sigma} e^{2 c \theta} \, .
\label{eq:dualclass7H}
\ee
The HKV $X_3 - \epsilon Z$ becomes the KV
\be
X_4 = c (u \partial_u - v \partial_v) - (\sigma / 2) \partial_\theta
\ee
and the dual spacetime admits an ${\cal S}_4 \supset {\cal H}_3 \supset {\cal G}_2$. The Lie bracket is
\[
[ k, X_4 ] = - c k \, .
\]
The spacetime is of isometry class 3 admitting an SCKV and HKV: Using the coordinate transformation $s = r \sin (\phi - \theta), t = r \cos (\phi - \theta)$
where $\phi = - \sigma \ln |u| / 2 c$ then $H = \delta (-1)^{\sigma - 2} u^{-2} (s^2 + t^2)^{-\sigma / 2} \exp [- 2 c \tan^{-1} (s / t) ]$.

The ${\cal S}_5$ specialization with $\sigma = 2$ is self-dual.

\subsection{Dual of isometry class 8}
The new function is $H= u^{-2} W(u^{-1}s) e^{-2 u^{-1} t}$.
The KV $X_2$ in (\ref{eq:KVsforcase8}) becomes an SCKV given by
\be
X_2 = u x^A \partial_A + u^2 \partial_u + \case{1}{2} x^A x_A \partial_v  \, ,
\label{eq:dualclass8SCKV}
\ee
and KV $X_3$ in (\ref{eq:KVsforcase8}) becomes an HKV given by
\be
\fl X_3 = - \epsilon \left( x^A \partial_A + u \partial_u + v \partial_v  \right)
- \rho \left( y \partial_v + u \partial_y \right)
- \sigma \left( z \partial_v + u \partial_z \right) \, , \qquad \psi = -\epsilon \, .
\label{eq:dualclass8HKV}
\ee
Thus we have ${\cal G}_3 \mapsto {\cal S}_3 \supset {\cal H}_2$. Thus the dual spacetime is an isometry class 1 specialization.

\subsection{Dual of isometry class 8 ($\epsilon = 0$)}

The new function is $H= u^{-2} W(u^{-1}s)$. We have ${\cal G}_4 \mapsto {\cal S}_4 \supset {\cal G}_3$. The KVs $X_3$ and $X_4$ remain KVs for the
new spacetime. For the $W(s) = K s^c$ specialization we have ${\cal H}_5 \mapsto {\cal S}_5 \supset {\cal H}_4$.

\subsection{Dual of isometry class 9}
The new metric function is $H= u^{-2} K \exp{ [2u^{-1}(\sigma z - \eta y)]}$.
The KV $X_5$ in (\ref{eq:KVcase9}) becomes the KV
\[
X_5 = \sigma \partial_y + \eta \partial_z,
\]
and the KV $-\eta X_3 - \sigma X_4$ becomes the KV
\[
-\eta X_3 - \sigma X_4 =  u \left( \eta \partial_z + \sigma\partial_y \right)
+ (\sigma y + \eta z ) \partial_v \, .
\]
The KV $X_2$ becomes the SCKV
\[
X_2 = u x^A \partial_A + u^2 \partial_u + \case{1}{2} x^A x_A \partial_v  \, ,
\]
and the KV $\sigma X_4 - \eta X_3$ becomes the HKV
\[
\fl \sigma X_4 - \eta X_3 =
2 \eta \sigma u \partial_u
+ (2 \eta \sigma v - \sigma y + \eta z) \partial_v
+ (2 \eta \sigma y - \sigma u) \partial_y
+ (2 \eta \sigma z + \eta u) \partial_z.
\]
The dual spacetime admits ${\cal S}_5 \supset {\cal H}_4$.

\subsection{Dual of isometry class 10}
The new dual spacetime also belongs to the isometry class 10 and the four KVs $X_i$ in (\ref{eq:KVcase10}) remain KVs and $Z$ remains
an HKV in the new spacetime. See reference \cite{tupper03} for details.

\section{Conclusions}

\label{sec:concl}
Given a particular isometry class one can in most instances say whether any other conformal symmetries are possible,
and one can certainly put an upper limit on the dimension of ${\cal C}$: For a general spacetime the dimension
of ${\cal S}$ is at most one greater than the dimension ${\cal H}$ and the maximum number of proper non-special
conformal Killing vectors in a type $N$ pp-wave spacetime is three.

We have identified three additional isometry classes not appearing in \cite{sippel86} and, since we do not impose the vacuum
condition from the outset, we have obtained spacetimes with conformal symmetries more general than in \cite{humberto83}.
The new isometry classes are class $1i$ (${\cal G}_3$) and class 8($\epsilon = 0$) (${\cal G}_4$) and their specializations.
Table \ref{tab:generalConfTransf} contains possibilities for further new isometry classes. However, these may correspond to known
isometry classes, either as they stand or under an appropriate coordinate transformation.

It is unknown whether there are any further spacetimes admitting the maximum number of three non-special CKV, other than
classes $A$, $B$ and $C$ in section \ref{sec:typeN}. However, we believe this to be unlikely. It is unknown whether the
specializations of isometry class 2 (in sections \ref{sec:typeN} and \ref{sec:SandG}),
class 3 (in sections  \ref{sec:typeN}, \ref{sec:SandG} and \ref{sec:typeNdual}) and class 4 (in sections  \ref{sec:typeN} and \ref{sec:SandG}) are the only possibilities.

\section*{Acknowledgments}

We would like to thank Graham Hall, Mohammed Patel, Malcolm MacCallum and Roy Maartens for useful discussions.

\end{document}